\documentclass[a4paper,USenglish,cleveref,thm-restate]{lipics-template/lipics-v2021}

\usepackage{booktabs}
\usepackage[table]{xcolor}
\usepackage{siunitx}
\usepackage{tikz}

\usetikzlibrary{calc, decorations.pathreplacing}
\usepackage{xspace}
\usepackage{enumitem}

\def\degseq{\ensuremath{\mathbf{d}}}
\def\rr#1{\ensuremath{\mathbf m_{h,l,d,t}(#1)}}
\def\rrr{\ensuremath{\mathbf m_{h,l,d,t}}}
\def\ie{i.e.,\xspace}
\def\eg{e.g.\xspace}
\def\etal{~et\,al.\xspace}

\def\ipwl{\textsc{Inc-Powerlaw}\xspace}
\def\igen{\textsc{Inc-Gen}\xspace}
\def\ireg{\textsc{Inc-Reg}\xspace}
\def\pld{\textsc{Pld}\xspace}
\def\parintra{\textsc{Intra-Run}\xspace}
\def\parinter{\textsc{Inter-Run}\xspace}
\def\Oh#1{\ensuremath{\mathcal O\!\left(#1\right)}}

\graphicspath{{./figures/}}

\hideLIPIcs

\title{Engineering Uniform Sampling of Graphs with a Prescribed Power-law Degree Sequence} 

\titlerunning{Uniform Sampling of Power-law Graphs} 

\author{Daniel Allendorf}{Goethe University Frankfurt, Germany}{dallendorf@ae.cs.uni-frankfurt.de}{}{}
\author{Ulrich Meyer}{Goethe University Frankfurt, Germany}{umeyer@ae.cs.uni-frankfurt.de}{}{}
\author{Manuel Penschuck}{Goethe University Frankfurt, Germany}{mpenschuck@ae.cs.uni-frankfurt.de}{}{}
\author{Hung Tran}{Goethe University Frankfurt, Germany}{htran@ae.cs.uni-frankfurt.de}{}{}
\author{Nick Wormald}{Monash University, Australia}{nicholas.wormald@monash.edu}{}{}

\authorrunning{D.~Allendorf, U.~Meyer, M.~Penschuck, H.~Tran and N.~Wormald}

\Copyright{Daniel Allendorf, Ulrich Meyer, Manuel Penschuck, Hung Tran and Nick Wormald}



\supplement{The source code of the graph generator is freely available at \url{https://github.com/Daniel-Allendorf/incpwl}.}

\funding{This work was partially supported by the Deutsche Forschungsgemeinschaft (DFG) under grants ME~2088/4-2 (SPP~1736 --- Algorithms for BIG DATA) and ME 2088/5-1 (FOR~2975 --- Algorithms, Dynamics, and Information Flow in Networks). It was also supported in part by the Australian Research Council under grant DP160100835.}

\begin{document}

\maketitle

\begin{abstract}
	We consider the following common network analysis problem:
	given a degree sequence $\degseq = (d_1, \dots, d_n) \in \mathbb N^n$ return a uniform sample from the ensemble of all simple graphs with matching degrees.
	In practice, the problem is typically solved using Markov Chain Monte Carlo approaches, such as Edge-Switching or Curveball, even if no practical useful rigorous bounds are known on their mixing times.
	In contrast, Arman\etal sketch \ipwl, a novel and much more involved algorithm capable of generating graphs for power-law bounded degree sequences with $\gamma \gtrapprox 2.88$ in expected linear time.

	For the first time, we give a complete description of the algorithm and add novel switchings.
	To the best of our knowledge, our open-source implementation of \ipwl is the first practical generator with rigorous uniformity guarantees for the aforementioned degree sequences.
	In an empirical investigation, we find that for small average-degrees \ipwl is very efficient and generates graphs with one million nodes in less than a second.
	For larger average-degrees, parallelism can partially mitigate the increased running-time.
\end{abstract}

\setcounter{page}{1}

\section{Introduction}
\label{sec:typesetting-summary}
A common problem in network science is the sampling of graphs matching prescribed degrees.
It is tightly related to the random perturbation of graphs while keeping their degrees.
Among other things, the problem appears as a building block in network models (\eg~\cite{Lancichinetti2009}).
It also yields null models used to estimate the statistical significance of observations (\eg~\cite{Milo824, gotelli1996null}).

The computational cost and algorithmic complexity of solving this problem heavily depend on the exact requirements.
Two relaxed variants with linear work sampling algorithms are Chung-Lu graphs~\cite{chung2002connected} and the configuration model~\cite{DBLP:journals/jct/BenderC78,DBLP:books/ox/Newman10,bollobas1985random}.
The Chung-Lu model produces the prescribed degree sequence only in expectation and allows for simple and efficient generators~\cite{DBLP:conf/waw/MillerH11,DBLP:conf/sc/AlamKVM16,MorenoPN18,DBLP:journals/jpdc/FunkeLMPSSSL19}.
The configuration model (see \autoref{sec:algorithm}), on the other hand, exactly matches the prescribed degree-sequence but allows loops and multi-edges, which introduce non-uniformity into the distribution~\cite[p.436]{DBLP:books/ox/Newman10} and are inappropriate for certain applications; however, erasing them may lead to significant changes in topology~\cite{DBLP:journals/snam/SchlauchHZ15,DBLP:journals/compnet/VigerL16}.

In this article, we focus on simple graphs (\ie without loops or multi-edges) matching a prescribed degree sequence exactly.

\subsection{Related Work}
An early uniform sampler with unknown algorithmic complexity was given by Tinhofer \cite{tinhofer1979generation}.
Perhaps the first practically relevant algorithm was implicitly given by graph enumeration methods (e.g.~\cite{bekessy1972asymptotic, DBLP:journals/jct/BenderC78, DBLP:journals/ejc/Bollobas80}) using the configuration model with rejection-sampling.
While its time complexity is linear in the number of nodes, it is exponential in the maximum degree squared and therefore already impractical for relatively small degrees.

McKay and Wormald \cite{DBLP:journals/jal/McKayW90} increased the permissible degrees.
Instead of repeatedly rejecting non-simple graphs, their algorithm may remove multi-edges using switching operations.
For $d$-regular graphs with $d = \Oh{n^{1/3}}$, its expected time complexity is $\Oh{d^3n}$ where $n$ is the number of nodes;
later, Gao and Wormald \cite{DBLP:conf/focs/GaoW15} improved the result to $d = o(\sqrt{n})$ with the same time complexity, and also considered sparse non-regular cases (e.g.~power-law degree sequences)~\cite{DBLP:conf/soda/GaoW18}.
Subsequently, Arman\etal~\cite{DBLP:conf/focs/ArmanGW19} present\footnote{Implementations of \igen and \ireg are available at\\ \url{https://users.monash.edu.au/~nwormald/fastgen_v3.zip}} the algorithms \igen, \ipwl and \ireg based on \emph{incremental relaxation}.
\igen runs in expected linear time provided $\Delta^4 = \Oh{m}$ where $\Delta$ is the maximum degree and $m$ is the number of edges.
\ireg reduces the expected time complexity to $\Oh{nd + d^4}$ if $d = o(\sqrt{n})$ for the regular case and \ipwl takes expected linear time for power-law degree sequences.

In the relaxed setting, where the generated graph is approximately uniform, Jerrum and Sinclair \cite{DBLP:journals/tcs/JerrumS90} gave an algorithm using Markov Chain Monte Carlo (MCMC) methods.
Since then, further MCMC-based algorithms have been proposed and analyzed (e.g.~\cite{DBLP:journals/cpc/CooperDG07, DBLP:conf/alenex/GkantsidisMMZ03, DBLP:conf/soda/Greenhill15, DBLP:journals/rsa/KannanTV99, DBLP:conf/sigcomm/MahadevanKFV06, DBLP:conf/alenex/StantonP11, strona2014fast, verhelst2008efficient, DBLP:journals/compnet/VigerL16}).
While these algorithms allow for larger families of degree sequences, topological restrictions (\eg connected graphs~\cite{DBLP:conf/alenex/GkantsidisMMZ03, DBLP:journals/compnet/VigerL16}), or more general characterizations (\eg joint degrees~\cite{DBLP:conf/alenex/StantonP11, DBLP:conf/sigcomm/MahadevanKFV06}), 
 theoretically proven upper bounds on their mixing times are either impractical or non-existent.
Despite this, some of these algorithms found wide use in several practical applications and have been implemented in freely available software libraries \cite{Lancichinetti2009, DBLP:journals/netsci/StaudtSM16, DBLP:journals/compnet/VigerL16} and adapted for advanced models of computation \cite{DBLP:conf/icpp/BhuiyanCKM14, DBLP:journals/jea/HamannMPTW18, DBLP:conf/esa/CarstensH0PTW18}.

As generally fast alternatives, asymptotic approximate samplers (e.g.~\cite{DBLP:conf/soda/GaoW18, DBLP:journals/algorithmica/BayatiKS10, DBLP:journals/combinatorica/KimV06, DBLP:journals/cpc/StegerW99, DBLP:journals/corr/Zhao13b}) have been proposed.
These samplers provide a weaker approximation than MCMC: the error tends to $0$ as $n$ grows but cannot be improved for any particular $n$.

\subsection{Our contribution}
Arman\etal~\cite{DBLP:conf/focs/ArmanGW19} introduce incremental relaxation and, as a corollary, obtain \ipwl by applying the technique to  the \pld algorithm \cite{DBLP:conf/soda/GaoW18}.
Crucial details of \ipwl were left open and are discussed here for the first time.
For the parts of the algorithm that use incremental relaxation (see \autoref{sec:algorithm}), we determine the order in which the relevant graph substructures should be relaxed, how to count the number of those substructures in a graph and find new lower bounds on the number of substructures, or adjust the ones used in \pld (see \autoref{sec:new-switchings}).

Our investigation also identified two cases where incremental relaxation compromised \ipwl's linear running-time as it implied too frequent restarts.
We solved this issue in consultation with the authors of \cite{DBLP:conf/focs/ArmanGW19} by adding new switchings to Phase 4 ($t_a$-, $t_b$-, and $t_c$-switchings, see \autoref{subsec:phase4}) and Phase 5 (switchings where $\max(m_1, m_2, m_3) = 2$, see \autoref{subsec:phase5}).

We engineer and optimize an \ipwl implementation and discuss practical parallelization possibilities.
In an empirical evaluation, we study our implementation's performance, provide evidence of its linear running-time, and compare the running-time with implementations of the popular approximately uniform \textsc{Edge-Switching} algorithm.

\subsection{Preliminaries and notation}
\label{subsec:notation}
For consistency, we use notation in accordance with prior descriptions of \pld and \ipwl.
A graph $G = (V, E)$ has $n$ nodes $V=\{1, \ldots, n\}$ and $|E|$ edges.
An edge connecting node $i$ to itself is called a \emph{loop} at $i$.
Let $m_{i,j}$ denote the multiplicity of edge $e = \{i, j\}$ (often abbreviated as $i j$);
for $m_{i,j} = 0, 1, 2, 3$ we refer to $e$ as a \emph{non-edge}, \emph{single-edge}, \emph{double-edge}, \emph{triple-edge}, respectively, and for $m_{i,j} > 1$ as \emph{multi-edge} (analogously for loops).
An edge is called \emph{simple} if it is neither a multi-edge nor loop.
A graph is \emph{simple} if it only contains simple edges (\ie no multi-edges or loops).

Given a graph $G$, define the \emph{degree} $\deg(i) = 2 m_{i,i} + \sum_{j \in V / \{i\}} m_{i,j}$ as the number of edges incident to node $i \in V$.
Let $\degseq = (d_1, \dots, d_n) \in \mathbb N ^ n$ be a \emph{degree sequence} and denote $\mathcal G(\degseq)$ as the set of simple graphs on $n$ nodes with $\deg(i) = d_i$ for all $i \in V$.
The degree sequence $\degseq$ is \emph{graphical} if $\mathcal G(\degseq)$ is non-empty.
If not stated differently, $\degseq$ is non-increasing, \ie $d_1 \ge d_2 \ge \ldots \ge d_n$.
We say $\degseq$ is \emph{power-law distribution-bounded (plib)} with exponent $\gamma > 1$ if $\degseq$ is strictly positive and there exists a constant $K$ (independent of $n$) such that for all $i \ge 1$ there are at most $Kni^{1-\gamma}$ entries of value $i$ or larger~\cite{DBLP:conf/soda/GaoW18}.
Denote the $k$-th factorial moment as $[x]_k = \prod_{i=0}^{k-1} (x-i)$ and define $M_k = \sum\nolimits_{i=1}^n [d_i]_k$, $H_k = \sum\nolimits_{i=1}^h [d_i]_k$, and $L_k = M_k - H_k$, where $h$ is a parameter defined in \autoref{subsec:p12preconditions} (roughly speaking, $h$ is the number of nodes with high degrees).
\textsl{See also \autoref{sec:apx-table-defs} for a summary of notation.}

\section{Algorithm description}
\label{sec:algorithm}
\setlength{\belowcaptionskip}{-10pt}
\ipwl takes a degree sequence $\degseq = (d_1, \dots, d_n)$ as input and outputs a uniformly random simple graph $G \in \mathcal{G}(\degseq)$.
The expected running-time is $O(n)$ if $\degseq$ is a plib sequence with $\gamma > 21/10 + \sqrt{61}/10 \approx 2.88102$.

The algorithm starts by generating a random graph $G$ using the \emph{configuration model}~\cite{DBLP:journals/jct/BenderC78}.
To this end, let $G$ be a graph with $n$ nodes and no edges, and for each node $i \in V$ place $d_i$ marbles labeled $i$ into an urn.
We then draw two random marbles without replacement, connect the nodes indicated by their labels, and repeat until the urn is empty.
The resulting graph $G$ is uniformly distributed in the set $\mathcal{S}(\rr{G})$ where $\rr{G}$ is a vector specifying the multiplicities  of all edges between, or loops at heavy nodes (as defined below), as well as the total numbers of other single-loops, double-edges, and triple-edges.
In particular, if $G$ is simple, then it is uniformly distributed in $\mathcal{G}(\degseq)$.
Moreover, if $\degseq$ implies $M_2 < M_1$, the degrees are rather small, and with constant probability $G$ is a simple graph~\cite{janson2009probability}.
Hence, \emph{rejection sampling} is efficient; the algorithm returns $G$ if it is simple and restarts otherwise.

For $M_2 \ge M_1$, the algorithm goes through five \emph{phases}.
In each phase, all non-simple edges of one kind, \eg all single-loops, or all double-edges, are removed from the graph by using \emph{switchings}.
A switching replaces some edges in the graph with other edges while preserving the degrees of all nodes.
Phases 1 and 2 remove multi-edges and loops with high multiplicity on the highest-degree nodes.
In Phases 3, 4 and 5, the remaining single-loops, triple-edges and double-edges are removed.
To guarantee the uniformity of the output and the linear running-time, the algorithm may restart in some steps.
A restart always resets the algorithm back to the first step of generating the initial graph.

Note that the same kind of switching can have different effects depending on which edges are selected for participation in the switching.
In general, we only allow the algorithm to perform switchings that have the intended effect.
Usually, a switching should remove exactly one non-simple edge without creating or removing other non-simple edges.
A switching that has the intended effect is called \emph{valid}.

Uniformity of the output is guaranteed by ensuring that the expected number of times a graph $G$ in $\mathcal{S}(\rr{G})$ is produced in the algorithm depends only on $\rr{G}$.
This requires some attention since, in general, the number of switchings we can perform on a graph and the number of switchings that produce a graph can vary between graphs in the same set (\ie some graphs are more likely reached than others).
To remedy this, there are \emph{rejection steps}, which restart the algorithm with a certain probability.
\emph{Before} a switching is performed on a graph $G$, the algorithm accepts with a probability proportional to the number of valid switchings that can be performed on $G$, and \emph{forward-rejects (f-rejects)} otherwise.
We do this by selecting an uniform random switching on $G$, and accepting if it is valid, or rejecting otherwise.
Then, \emph{after} a switching produced a graph $G'$, the algorithm accepts with a probability inversely proportional to the number of valid switchings that can produce $G'$, and \emph{backward-rejects (b-rejects)} otherwise.
This is done by computing a quantity $b(G')$ that is proportional to the number of valid switchings that can produce $G'$, and a lower bound $\underline{b}(G')$ on $b(G')$ over all $G'$ in the same set, and then accepting with probability $\underline{b}(G')/b(G')$.

\subsection{Phase 1 and 2 preconditions}
\label{subsec:p12preconditions}
In Phases 1 and 2, the algorithm removes non-simple edges with high multiplicity between the highest-degree nodes.
To this end, define a parameter $h = n^{1 - \delta(\gamma - 1)}$ where $\delta$ is chosen so that $1/(2 \gamma - 3) < \delta < (2 - 3/(\gamma - 1))/(4 - \gamma)$ (e.g. $\delta \approx 0.362$ for $\gamma \approx 2.88103$).
The $h$ highest-degree nodes are then called \emph{heavy}, and the remaining nodes are called \emph{light}.
An edge is called \emph{heavy} if its incident nodes are heavy, and \emph{light} otherwise.
A \emph{heavy multi-edge} is a multi-edge between heavy nodes, and a \emph{heavy loop} is a loop at a heavy node.

Now, let $W_i$ denote the sum of the multiplicities of all heavy multi-edges incident with $i$, and let $W_{i,j} = W_i + 2 m_{i,i} - m_{i,j}$.
Finally, let $\eta = \sqrt{M_2^2 H_1 / M_1^3}$. 
There are four preconditions for Phase 1 and 2: (1) for all nodes $i \neq j$ connected by a heavy multi-edge, we have $m_{i, j} W_{i, j} \leq \eta d_i$ and $m_{i, j} W_{j, i} \leq \eta d_j$, (2) for all nodes $i$ that have a heavy loop, we have $m_{i,i} W_i \leq \eta d_i$, (3) the sum of the multiplicities of all heavy multi-edges is at most $4 M_2^2 / M_1^2$, and (4) the sum of the multiplicities of all heavy loops is at most $4 M_2 / M_1$.

If any of the preconditions is not met, the algorithm restarts, otherwise it enters Phase 1.

\subsection{Phase 1: removal of heavy multi-edges}
\label{subsec:phase1}
A heavy multi-edge $i j$ with multiplicity $m = m_{i,j}$ is removed with the \emph{heavy-$m$-way switching} shown in \autoref{fig:heavymwayswitching}.
Note that the switching is defined on \emph{pairs} instead of edges.
An edge $i j$ of multiplicity $m$ is treated as $m$ distinct pairs $(i, j)$.
Adding a pair $(i, j)$ increases the multiplicity $m$, and similarly, removing $(i, j)$ decreases $m$.
The heavy-$m$-way switching switching removes the $m$ pairs $(i, j)$ and $m$ additional pairs $(v_k, v_{k+1}), 1 \leq k \leq m$, and replaces them with $2m$ new pairs between $i$ and $v_k$, and $j$ and $v_{k+1}$.

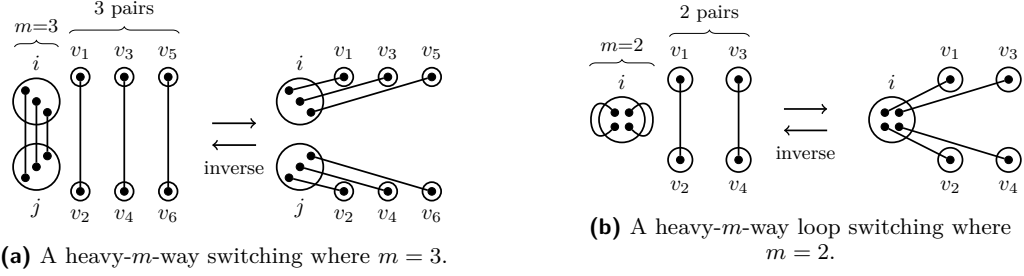
\begin{figure}[t]
 \begin{subfigure}{.45\textwidth}
 	\resizebox{\textwidth}{!}{
 		\def\node#1{}
\newcommand{\drawVerticesWithXOffsetPrefix}[2]{
	\def\y{1.5em};
	\def\x{2em}
	\def\px{.5em}
	\node[point] (#2i1) at (#1em,       \y)        {};
	\node[point] (#2i2) at (#1em + \px, \y - .5em) {};
	\node[point] (#2i3) at (#1em - \px, \y + .5em) {};
	\node[heavyvertex, label=above:{$i$}] (#2i) at (#1em, \y) {};
	\node[point] (#2j1) at (#1em,       -\y)        {};
	\node[point] (#2j2) at (#1em + \px, -\y + .5em) {};
	\node[point] (#2j3) at (#1em - \px, -\y - .5em) {};
	\node[heavyvertex, label=below:{$j$}] (#2j) at (#1em, -\y) {};
	\def\my{1.75}
	\node[point] (#2v1) at (#1em + \x,  \my*\y) {};
	\node[point] (#2v2) at (#1em + \x, -\my*\y) {};
	\node[nonheavyvertex, label=above:{$v_1$}] at (#1em + \x,  \my*\y) {};
	\node[nonheavyvertex, label=below:{$v_2$}] at (#1em + \x, -\my*\y) {};
	\node[point] (#2v3) at (#1em + 2*\x,  \my*\y) {};
	\node[point] (#2v4) at (#1em + 2*\x, -\my*\y) {};
	\node[nonheavyvertex, label=above:{$v_3$}] at (#1em + 2*\x,  \my*\y) {};
	\node[nonheavyvertex, label=below:{$v_4$}] at (#1em + 2*\x, -\my*\y) {};
	\node[point] (#2v5) at (#1em + 3*\x,  \my*\y) {};
	\node[point] (#2v6) at (#1em + 3*\x, -\my*\y) {};
	\node[nonheavyvertex, label=above:{$v_5$}] at (#1em + 3*\x,  \my*\y) {};
	\node[nonheavyvertex, label=below:{$v_6$}] at (#1em + 3*\x, -\my*\y) {};
}
\begin{tikzpicture}[
	point/.style={
		draw,
		circle,
		inner sep=0,
		fill,
		minimum width=0.35em,
		minimum height=0.35em
	},
	heavyvertex/.style={
		draw, 
		circle,
		inner sep=0.75em, 
		thick
	},
	nonheavyvertex/.style={
		draw, 
		circle,
		inner sep=0.3em, 
		thick
	},
	edge/.style={
		draw, 
		black, 
		solid,
		thick
	}]
	
	\drawVerticesWithXOffsetPrefix{-3}{};
	\drawVerticesWithXOffsetPrefix{9}{c};
	
	\path[->, thick, draw] (-3em + 4*\x,  0.5em) -- (-3em + 5*\x,  0.5em) {};
	\path[->, thick, draw] (-3em + 5*\x, -0.5em) -- (-3em + 4*\x, -0.5em) node [midway, label=below:{\small inverse}] {};
	
	\path[edge] (i1) -- (j1);
	\path[edge] (i2) -- (j2);
	\path[edge] (i3) -- (j3);
	\path[edge] (v1) -- (v2);
	\path[edge] (v3) -- (v4);
	\path[edge] (v5) -- (v6);
	\draw [decorate, decoration={brace, amplitude=.3em,mirror,raise=2em}, yshift=-3em] ($(v5) + (.5em,0)$) -- ($(v1) + (-.5em,0)$) node [midway, yshift=2em, label=above:{\small $3$ pairs}] {};
	\draw [decorate, decoration={brace, amplitude=.3em,mirror,raise=2em}, yshift=-3em] ($(i) + (1em,.5em)$) -- ($(i) + (-1em,.5em)$) node [midway, yshift=2em, label=above:{\small $m{=}3$}] {};
	
	\path[edge] (ci3) -- (cv1);
	\path[edge] (cj3) -- (cv2);
	\path[edge] (ci1) -- (cv3);
	\path[edge] (cj1) -- (cv4);
	\path[edge] (ci2) -- (cv5);
	\path[edge] (cj2) -- (cv6);
\end{tikzpicture}
	}
 	\caption{\centering A heavy-$m$-way switching where $m = 3$.}
 	\label{fig:heavymwayswitching}
 \end{subfigure}
\hfill
 \begin{subfigure}{.45\textwidth}
 	\resizebox{\textwidth}{!}{
 		\def\node#1{}
\newcommand{\drawVerticesWithXOffsetPrefix}[2]{
	\def\y{1.5em};
	\def\x{1.35em}
	\def\px{.33em}
	\node[point] (#2i1) at (#1em - \px, -\px) {};
	\node[point] (#2i2) at (#1em + \px, -\px) {};
	\node[point] (#2i3) at (#1em - \px, +\px) {};
	\node[point] (#2i4) at (#1em + \px, +\px) {};
	\node[heavyvertex, label=above:{$i$}] (#2i) at (#1em, 0) {};
	\def\my{1.25}
	\node[point] (#2v1) at (#1em + 2*\x,  \my*\y) {};
	\node[point] (#2v2) at (#1em + 2*\x, -\my*\y) {};
	\node[nonheavyvertex, label=above:{$v_1$}] at (#1em + 2*\x,  \my*\y) {};
	\node[nonheavyvertex, label=below:{$v_2$}] at (#1em + 2*\x, -\my*\y) {};
	\node[point] (#2v3) at (#1em + 4*\x,  \my*\y) {};
	\node[point] (#2v4) at (#1em + 4*\x, -\my*\y) {};
	\node[nonheavyvertex, label=above:{$v_3$}] at (#1em + 4*\x,  \my*\y) {};
	\node[nonheavyvertex, label=below:{$v_4$}] at (#1em + 4*\x, -\my*\y) {};
}
\begin{tikzpicture}[
	point/.style={
		draw,
		circle,
		inner sep=0,
		fill,
		minimum width=0.35em,
		minimum height=0.35em
	},
	heavyvertex/.style={
		draw, 
		circle,
		inner sep=0.75em, 
		thick
	},
	nonheavyvertex/.style={
		draw, 
		circle,
		inner sep=0.4em, 
		thick
	},
	edge/.style={
		draw, 
		black, 
		solid,
		thick
	}]
	
	\drawVerticesWithXOffsetPrefix{-3.5}{};
	\drawVerticesWithXOffsetPrefix{9}{c};
	
	\path[->, thick, draw] (5em - 1em,  0.5em) -- (5em + 1em,  0.5em) {};
	\path[->, thick, draw] (5em + 1em, -0.5em) -- (5em - 1em, -0.5em) node [midway,label=below:{\small inverse}] {};
	
	\draw[edge] (i3) to [out=135, in=-135, looseness=5] (i1);
	\draw[edge] (i4) to [out= 45, in= -45, looseness=5] (i2);
	\path[edge] (v1) -- (v2);
	\path[edge] (v3) -- (v4);
	\draw [decorate, decoration={brace, amplitude=.3em,mirror,raise=2em}, yshift=-3em] ($(v3) + (.5em,0)$) -- ($(v1) + (-.5em,0)$) node [midway, yshift=2em, label=above:{\small $2$ pairs}] {};
	\draw [decorate, decoration={brace, amplitude=.3em,mirror,raise=2em}, yshift=-3em] ($(i) + (1.5em,.5em)$) -- ($(i) + (-1.5em,.5em)$) node [midway, yshift=2em, label=above:{\small $m{=}2$}] {};
	
	\path[edge] (ci3) -- (cv1);
	\path[edge] (ci1) -- (cv2);
	\path[edge] (ci4) -- (cv3);
	\path[edge] (ci2) -- (cv4);
\end{tikzpicture}
 	}
 	\caption{\centering A heavy-$m$-way loop switching where $m = 2$.}
 	\label{fig:heavymwayloopswitching}
 \end{subfigure}
 \vspace{1em}
 \caption{\centering Switchings used in Phases 1 and 2.}
\end{figure}

In Phase 1, we iterate over all heavy multi-edges $ij$ and each time execute:

\begin{enumerate}
\item Pick a uniform random heavy-$m$-way switching $S = (G, G')$ at nodes $i$ and $j$ as follows:
for all $1 \leq k \leq m$, sample a uniform random pair $(v_k, v_{k+1})$ in random orientation.
Then remove the pairs $(i, j)$ and $(v_k, v_{k+1})$, and add $(i, v_k)$ and $(j, v_{k+1})$.
The graph that results after all pairs have been switched is $G'$.
 \item Restart the algorithm (f-reject) if $S$ is not valid.
The switching is valid if for all $1 \leq k \leq m$: (a) $v_k$ and $v_{k+1}$ are distinct from $i$ and $j$, (b) if $v_k$ is heavy, it is not already connected to $i$, and if $v_{k+1}$ is heavy, it is not connected to $j$, and (c) at least one of $v_k$ and $v_{k+1}$ is light.
(This ensures that only the heavy multi-edge $i j$ is removed, and no other heavy multi-edges or loops are added or removed.)
 \item Restart the algorithm (b-reject) with probability $1 - \underline{b}_{hm}(G', i, j, m) / b_{hm}(G', i, j, m)$.
 \item With probability $1 / (1 + \overline{b}_{hm}(G', i, j, 1) / \underline{f}_{hm}(G', i, j, 1))$, set $G \gets G'$ and continue to the next iteration.
Otherwise, re-add $i j$ as a single-edge with the following steps:
 \begin{enumerate}[nosep]
  \item Pick a uniform random inverse heavy-$1$-way switching $S' = (G', G'')$ at nodes $i, j$ as follows:
	pick one simple neighbor $v_1$ of $i$ (\ie edge $v_1i$ is simple) uniformly at random, and analogously $v_2$ for $j$.
	Then remove the pairs $(i, v_1)$, $(j, v_2)$, and add $(i, j)$, $(v_1, v_2)$.
  \item Restart the algorithm (f-reject) if $S'$ is not valid. The switching is valid unless both $v_1$ and $v_2$ are heavy.
  \item Restart the algorithm (b-reject) with probability $1 - \underline{f}_{hm}(G'', i, j, 1) / f_{hm}(G'', i, j, 1)$.
  \item Set $G \gets G''$.
 \end{enumerate}
\end{enumerate}

\noindent
To compute the b-rejection probability for step~3, let $Y_1$ and $Y_2$ be the number of heavy nodes that are neighbors of $i$ and $j$, respectively, in the graph $G'$.
Then, set:
\begin{align}
 \underline{b}_{hm}(G', i, j, m) &= [d_i {-} W_{i,j}]_m [d_j {-} W_{j,i}]_m - m h^2 [d_i {-} W_{i,j}]_{m-1} [d_j {-} W_{j,i}]_{m-1}
 \\
 b_{hm}(G', i, j, m) &= \sum_{l=0}^m \bigg[  (-1)^{l} \binom{m}{l} [Y_1]_l [Y_2]_l [d_i {-} W_{i,j} {-} l]_{m - l} [d_j {-} W_{j,i} {-} l]_{m - l} \bigg]
\end{align}

For step $4$:
\begin{align}
 \overline{b}_{hm}(G', i, j, 1) &= (d_i - W_{i,j}) (d_j - W_{j,i})
 \\
 \underline{f}_{hm}(G', i, j, 1) &= M_1 - 2 H_1
\end{align}

For step $4c$: let $Z_1$ be the number of ordered pairs between light nodes in the graph $G''$, let $Z_2$ be the number of pairs between one light and one heavy node, where the heavy node is not adjacent to $i$, and let $Z_3$ be the analogous number for $j$.
Then set:
\begin{align}
  f_{hm}(G'', i, j, 1) &= Z_1 + Z_2 + Z_3 
  \\
  \underline{f}_{hm}(G'', i, j, 1) &= M_1 - 2 H_1
\end{align}

Phase 1 ends if all heavy multi-edges are removed. Then \ipwl enters Phase 2.

\subsection{Phase 2: removal of heavy loops}
\label{subsec:phase2}
Phase 2 removes all heavy loops using the \emph{heavy-m-way loop switching} shown in \autoref{fig:heavymwayloopswitching}.
The algorithm iterates over all heavy nodes $i$ that have a heavy loop, and for each performs the following steps:

\begin{enumerate}
 \item Pick a uniform random heavy-$m$-way loop switching $S = (G, G')$ at node $i$ (cf. Phase 1).
 \item Restart (f-reject) if $S$ is not valid.
 The switching is valid if for all $1 \leq k \leq m$: 
 a) $v_k \ne i$ and $v_{k+1} \ne i$,
 b) $iv_k$ and $iv_{k+1}$ are non-edges or light, c) at least one of $v_k$ and $v_{k+1}$ is light.
 \item Restart (b-reject) with probability $1 - \underline{b}_{hl}(G', i, m) / b_{hl}(G', i, m)$.
 \item Set $G \gets G'$.
\end{enumerate}

\noindent Let $Y$ be the number of heavy neighbors of $i$ in $G'$.
The quantities needed in step $3$ are:
\begin{align}
 \underline{b}_{hl}(G', i, m) &= [d_i]_{2m} - m h^2 [d_i]_{2m-2}
 \\
 b_{hl}(G', i, m) &= \sum_{l=0}^m (-1)^{l} \binom{m}{l} [Y]_{2l} [d_i - 2l]_{2m - 2l}
\end{align}

Phase 2 ends if all heavy loops are removed. We then check preconditions for the next phases.

\subsection{Phase 3, 4 and 5 preconditions}
\label{subsec:phase3prec}
After Phases 1 and 2, the only remaining non-simple edges in the graph $G$ are all incident with at least one light node, \ie with one of the low-degree nodes.
With constant probability, the only remaining non-simple edges are single loops, double-edges, and triple-edges, and there are not too many of them \cite{DBLP:conf/soda/GaoW18}.
Otherwise, the algorithm restarts.
Let $m_l$ denote the number of single loops, $m_t$ the number of triple-edges, and $m_d$ the number of double-edges in the graph $G$.
Then, the preconditions are: (1) $m_l \leq 4 L_2 / M_1$, (2) $m_t \leq 2 L_3 M_3 / M_1^3$, (3) $m_d \leq 4 L_2 M_2 / M_1^2$, and (4) there are no loops or multi-edges of higher multiplicity.

If all preconditions are met, the algorithm enters Phase 3 to remove all remaining loops.

\subsection{Phase 3: removal of light loops}
\label{subsec:phase3}
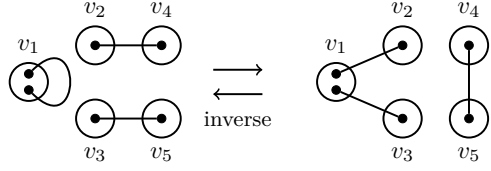
\begin{figure}[t]
	\centering
	\resizebox{0.49\textwidth}{!}{
		\def\node#1{}
\newcommand{\drawVerticesWithXOffsetPrefix}[2]{
	\def\y{1.5em};
	\def\x{1.35em}
	\def\px{.33em}
	\node[point] (#2i1) at (#1em, +\px) {};
	\node[point] (#2i2) at (#1em, -\px) {};
	\node[2pointvertex, label=above:{$v_1$}] (#2i) at (#1em, 0) {};
	\def\my{1}
	\node[point] (#2v1) at (#1em + 2*\x,  \my*\y) {};
	\node[point] (#2v2) at (#1em + 2*\x, -\my*\y) {};
	\node[2pointvertex, label=above:{$v_2$}] at (#1em + 2*\x,  \my*\y) {};
	\node[2pointvertex, label=below:{$v_3$}] at (#1em + 2*\x, -\my*\y) {};
	\node[point] (#2v3) at (#1em + 4*\x,  \my*\y) {};
	\node[point] (#2v4) at (#1em + 4*\x, -\my*\y) {};
	\node[2pointvertex, label=above:{$v_4$}] at (#1em + 4*\x,  \my*\y) {};
	\node[2pointvertex, label=below:{$v_5$}] at (#1em + 4*\x, -\my*\y) {};
}
\begin{tikzpicture}[
	point/.style={
		draw,
		circle,
		inner sep=0,
		fill,
		minimum width=0.35em,
		minimum height=0.35em
	},
	heavyvertex/.style={
		draw, 
		circle,
		inner sep=0.75em, 
		thick
	},
	nonheavyvertex/.style={
		draw, 
		circle,
		inner sep=0.4em, 
		thick
	},
	2pointvertex/.style={
		draw,
		circle,
		inner sep=0.55em,
		thick
	},
	edge/.style={
		draw, 
		black, 
		solid,
		thick
	}]
	
	\drawVerticesWithXOffsetPrefix{-3.5}{};
	\drawVerticesWithXOffsetPrefix{9}{c};
	
	\path[->, thick, draw] (5em - 1em,  0.5em) -- (5em + 1em,  0.5em) {};
	\path[->, thick, draw] (5em + 1em, -0.5em) -- (5em - 1em, -0.5em) node [midway,label=below:{\small inverse}] {};
	
	\draw[edge] (i1) to [out=45, in=-45, looseness=8] (i2);
	\path[edge] (v1) -- (v3);
	\path[edge] (v2) -- (v4);
	
	\path[edge] (ci1) -- (cv1);
	\path[edge] (ci2) -- (cv2);
	\path[edge] (cv3) -- (cv4);
\end{tikzpicture}
	}
	\caption{\centering The $l$-switching used in Phase 3.}
	\label{fig:lswitching}
\end{figure}

Phase 3 removes all light loops, \ie loops at lower degree nodes, with the \emph{l-switching} depicted in \autoref{fig:lswitching}.
We repeat the following steps until all loops are removed:

\begin{enumerate}
  \item Pick a uniform random $l$-switching $S = (G, G')$ as follows.
Sample a uniform random loop on some node $v_1$ in $G$.
Then, sample two uniform random pairs $(v_2, v_4)$ and $(v_3, v_5)$ in random orientation.
Replace $(v_1, v_1)$, $(v_2, v_4)$, $(v_3, v_5)$ with $(v_1, v_2)$, $(v_1, v_3)$, $(v_4, v_5)$.
 \item Restart (f-reject) if $S$ is not valid.
The switching is valid if it removes the targeted loop without adding or removing other multi-edges or loops.
 \item Restart (b-reject) with probability $1 - \underline{b}_l(G';0) \underline{b}_l(G';1) / (b_l(G', \emptyset) b_l(G', v_1 v_2 v_3))$.
 \item Set $G \gets G'$.
\end{enumerate}

\noindent
To accelerate the computation of the b-rejection probabilities, \emph{incremental relaxation}~\cite{DBLP:conf/focs/ArmanGW19} is used.
Let $v_1 v_2 v_3$ denote a \emph{two-star} centered at $v_1$, i.e. three nodes $v_1, v_2, v_3$ where $v_1 v_2$ and $v_1 v_3$ are edges.
We call a two-star $v_1 v_2 v_3$ \emph{simple}, if both edges  are simple, and we call the star \emph{light}, if the center $v_1$ is a light node.
Finally, we speak of \emph{ordered} two-stars if each permutation of the labels for the outer nodes $v_2$ and $v_3$ implies a distinct two-star.
Then, the $l$-switching creates a light simple two-star $v_1 v_2 v_3$ and a simple pair $v_4 v_5$.

With incremental relaxation, the b-rejection is split up into two sub-rejections, one for each structure created by the switching.
First, set $b_l(G', \emptyset)$ to the number of light simple ordered two-stars in $G'$.
Then, initialize $b_l(G', v_1 v_2 v_3)$ to the number of simple ordered pairs in $G'$.
Now, subtract all the simple ordered pairs that are incompatible with the two-star $v_1 v_2 v_3$ created by the switching.
The incompatible pairs a) share nodes with the two-star $v_1 v_2 v_3$ or b) have edges $v_2 v_4$ or $v_3 v_5$.
Let $A_2 = \sum\nolimits_{i=1}^{d_1} d_i$.
Then, we use the following lower bounds on these quantities:
\begin{align}
 \underline{b}_l(G';0) &= L_2 - 12 m_t d_h - 8 m_d d_h - m_l d_h^2
 \\
 \underline{b}_l(G';1) &= M_1 - 6 m_t - 4 m_d - 2 m_l - 2 A_2 - 4 d_1 - 2 d_h
\end{align}

Next, the algorithm removes the triple-edges in Phase~4.

\subsection{Phase 4: removal of light triple-edges}
\label{subsec:phase4}
In Phase 4, the algorithm uses multiple different switchings.
Similarly to the previous phases, there is one switching that removes the multi-edges.
The other switchings, called \emph{boosters}, lower the probability of a b-rejection.
In total, there are four different switchings. The $t$-switching removes a triple-edge (see \autoref{fig:tswitching}).
The $t_a$-, $t_b$- and $t_c$-switchings create structures consisting of a simple three-star $v_1 v_3 v_5 v_7$, and a light simple three star $v_2 v_4 v_6 v_8$, that do not share any nodes.
We call these structures \emph{triplets}.
Note that the t-switching creates a triplet where none of the edges $v_1 v_2$, $v_3 v_4$, $v_5 v_6$ or $v_7 v_8$ are allowed to exist.
The $t_a$-switching creates the triplets where either one of the edges $v_3 v_4$, $v_5 v_6$ or $v_7 v_8$ exist.
The $t_b$-switching (see \autoref{fig:tbswitching}) creates the triplets where two of those edges exist.
The $t_c$-switching creates the triplet where all three of those edges exist.

\begin{figure}[t]
 \begin{subfigure}{.4\textwidth}
 	\resizebox{\textwidth}{!}{
		\def\node#1{}
\newcommand{\drawVerticesWithXOffsetPrefix}[2]{
	\def\y{1.5em};
	\def\x{2em}
	\def\spx{.4em}
	\def\px{.5em}
	\node[smallpoint] (#2i1) at (#1em - \spx, \y) {};
	\node[smallpoint] (#2i2) at (#1em,       \y) {};
	\node[smallpoint] (#2i3) at (#1em + \spx, \y) {};
	\node[2pointvertex, label=above:{$v_1$}] (#2i) at (#1em, \y) {};
	\node[smallpoint] (#2j1) at (#1em - \spx, -\y) {};
	\node[smallpoint] (#2j2) at (#1em,       -\y) {};
	\node[smallpoint] (#2j3) at (#1em + \spx, -\y) {};
	\node[2pointvertex, label=below:{$v_2$}] (#2j) at (#1em, -\y) {};
	\node[point] (#2v3) at (#1em + \x,  \y) {};
	\node[point] (#2v4) at (#1em + \x, -\y) {};
	\node[2pointvertex, label=above:{$v_3$}] at (#1em + \x,  \y) {};
	\node[2pointvertex, label=below:{$v_4$}] at (#1em + \x, -\y) {};
	\coordinate (#2v5coords) at (#1em - 2*\x,  \y);
	\coordinate (#2v6coords) at (#1em - 2*\x, -\y);
	\node[point] (#2v5) at (#2v5coords) {};
	\node[point] (#2v6) at (#2v6coords) {};
	\node[2pointvertex, label=above:{$v_5$}] at (#2v5coords) {};
	\node[2pointvertex, label=below:{$v_6$}] at (#2v6coords) {};
	\coordinate (#2v7coords) at (#1em - \x,  3*\y);
	\coordinate (#2v8coords) at (#1em - \x, -3*\y);
	\node[point] (#2v7) at (#2v7coords) {};
	\node[point] (#2v8) at (#2v8coords) {};
	\node[2pointvertex, label=above:{$v_7$}] at (#2v7coords) {};
	\node[2pointvertex, label=below:{$v_8$}] at (#2v8coords) {};
}
\begin{tikzpicture}[
	point/.style={
		draw,
		circle,
		inner sep=0,
		fill,
		minimum width=0.35em,
		minimum height=0.35em
	},
	smallpoint/.style={
		draw,
		circle,
		inner sep=0,
		fill,
		minimum width=0.28em,
		minimum height=0.28em,
	},
	heavyvertex/.style={
		draw, 
		circle,
		inner sep=0.75em, 
		thick
	},
	nonheavyvertex/.style={
		draw, 
		circle,
		inner sep=0.4em, 
		thick
	},
	2pointvertex/.style={
		draw,
		circle,
		inner sep=0.55em,
		thick
	},
	edge/.style={
		draw, 
		black, 
		solid,
		thick
	}]
	
	\drawVerticesWithXOffsetPrefix{-3}{};
	\drawVerticesWithXOffsetPrefix{13}{c};
	
	\path[->, thick, draw] (-1em + 2*\x,  0.5em) -- (-1em + 3*\x,  0.5em) {};
	\path[->, thick, draw] (-1em + 3*\x, -0.5em) -- (-1em + 2*\x, -0.5em) node [midway, label=below:{\small inverse}] {};
	
	\path[edge] (i1) -- (j1);
	\path[edge] (i2) -- (j2);
	\path[edge] (i3) -- (j3);
	\path[edge] (v3) -- (v4);
	\path[edge] (v5) -- (v6);
	\path[edge] (v7) -- (v8);
	
	\path[edge] (ci3) -- (cv3);
	\path[edge] (cj3) -- (cv4);
	\path[edge] (ci1) -- (cv5);
	\path[edge] (cj1) -- (cv6);
	\path[edge] (ci2) -- (cv7);
	\path[edge] (cj2) -- (cv8);
	
	\coordinate (cv9)  at (9em - 2.25*\x,  3.5*\y);
	\coordinate (cv10) at (9em - 2.25*\x, -3.5*\y););
	\coordinate (cv11) at (9em +  .75*\x,  3.5*\y);
	\coordinate (cv12) at (9em +  .75*\x, -3.5*\y););
	\node[label=above:{\phantom{$v_9$}}] (v9phantom) at (cv9) {};
	\draw[draw, white, solid, thick] (cv9)  to [out=75, in=105, looseness=1.05]   (cv11);
	\draw[draw, white, solid, thick] (cv10) to [out=-75, in=-105, looseness=1.05] (cv12);
\end{tikzpicture}
	}
	\caption{\centering The $t$-switching.}
	\label{fig:tswitching}
 \end{subfigure}
\hfill
 \begin{subfigure}{.55\textwidth}
 	\resizebox{\textwidth}{!}{
		\def\node#1{}
\newcommand{\drawVerticesWithXOffsetPrefix}[2]{
	\def\y{1.5em};
	\def\x{2em}
	\def\spx{.4em}
	\def\px{.5em}
	\node[smallpoint] (#2v1p1) at (#1em - \spx, \y) {};
	\node[smallpoint] (#2v1p2) at (#1em,       \y) {};
	\node[smallpoint] (#2v1p3) at (#1em + \spx, \y) {};
	\node[2pointvertex, label=above:{$v_1$}] (#2v1) at (#1em, \y) {};
	\coordinate (#2v7coords) at (#1em - \x,  3*\y);
	\node[point] (#2v7) at (#2v7coords) {};
	\node[2pointvertex, label=above:{$v_7$}] at (#2v7coords) {};
	\coordinate (#2v9coords) at (#1em - 2.25*\x,  3.5*\y);
	\node[point] (#2v9) at (#2v9coords) {};
	\node[2pointvertex, label=above:{$v_9$}] at (#2v9coords) {};
	\coordinate (#2v11coords) at (#1em + .75*\x,  3.5*\y);
	\node[point] (#2v11) at (#2v11coords) {};
	\node[2pointvertex] at (#2v11coords) {};
	
	\node[smallpoint] (#2v2p1) at (#1em - \spx, -\y) {};
	\node[smallpoint] (#2v2p2) at (#1em,       -\y) {};
	\node[smallpoint] (#2v2p3) at (#1em + \spx, -\y) {};
	\node[2pointvertex, label=below:{$v_2$}] (#2v2) at (#1em, -\y) {};
	\coordinate (#2v8coords) at (#1em - \x, -3*\y);
	\node[point] (#2v8) at (#2v8coords) {};
	\node[2pointvertex, label=below:{$v_8$}] at (#2v8coords) {};
	\coordinate (#2v10coords) at (#1em - 2.25*\x,  -3.5*\y);
	\node[point] (#2v10) at (#2v10coords) {};
	\node[2pointvertex] at (#2v10coords) {};
	\coordinate (#2v12coords) at (#1em + .75*\x,  -3.5*\y);
	\node[point] (#2v12) at (#2v12coords) {};
	\node[2pointvertex] at (#2v12coords) {};
	
	\node[point] (#2v3p1) at (#1em + \x - .5*\px,  \y) {};
	\node[point] (#2v3p2) at (#1em + \x + .5*\px,  \y) {};
	\node[2pointvertex, label=above:{$v_3$}] at (#1em + \x,  \y) {};
	\coordinate (#2v13coords) at (#1em + 2.25*\x, 3*\y);
	\node[point] (#2v13) at (#2v13coords) {};
	\node[2pointvertex] at (#2v13coords) {};
	\coordinate (#2v15coords) at (#1em + 2.75*\x, \y);
	\node[point] (#2v15) at (#2v15coords) {};
	\node[2pointvertex] at (#2v15coords) {};
	
	\node[point] (#2v4p1) at (#1em + \x - .5*\px, -\y) {};
	\node[point] (#2v4p2) at (#1em + \x + .5*\px, -\y) {};
	\node[2pointvertex, label=below:{$v_4$}] at (#1em + \x, -\y) {};
	\coordinate (#2v14coords) at (#1em + 2.25*\x, -3*\y);
	\node[point] (#2v14) at (#2v14coords) {};
	\node[2pointvertex] at (#2v14coords) {};
	\coordinate (#2v16coords) at (#1em + 2.75*\x, -\y);
	\node[point] (#2v16) at (#2v16coords) {};
	\node[2pointvertex] at (#2v16coords) {};
	
	\coordinate (#2v5coords) at (#1em - 2*\x,  \y);
	\node[point] (#2v5p1) at (#1em - 2*\x - .5*\px, \y) {};
	\node[point] (#2v5p2) at (#1em - 2*\x + .5*\px, \y) {};
	\node[2pointvertex, label=above:{$v_5$}] at (#2v5coords) {};
	\coordinate (#2v17coords) at (#1em - 3.5*\x, 3*\y);
	\node[point] (#2v17) at (#2v17coords) {};
	\node[2pointvertex] at (#2v17coords) {};
	\coordinate (#2v19coords) at (#1em - 4*\x, \y);
	\node[point] (#2v19) at (#2v19coords) {};
	\node[2pointvertex] at (#2v19coords) {};
	
	\coordinate (#2v6coords) at (#1em - 2*\x, -\y);
	\node[point] (#2v6p1) at (#1em - 2*\x - .5*\px, -\y) {};
	\node[point] (#2v6p2) at (#1em - 2*\x + .5*\px, -\y) {};
	\node[2pointvertex, label=below:{$v_6$}] at (#2v6coords) {};
	\coordinate (#2v18coords) at (#1em - 3.5*\x, -3*\y);
	\node[point] (#2v18) at (#2v18coords) {};
	\node[2pointvertex] at (#2v18coords) {};
	\coordinate (#2v20coords) at (#1em - 4*\x, -\y);
	\node[point] (#2v20) at (#2v20coords) {};
	\node[2pointvertex] at (#2v20coords) {};
}
\begin{tikzpicture}[
	point/.style={
		draw,
		circle,
		inner sep=0,
		fill,
		minimum width=0.35em,
		minimum height=0.35em
	},
	smallpoint/.style={
		draw,
		circle,
		inner sep=0,
		fill,
		minimum width=0.28em,
		minimum height=0.28em,
	},
	2pointvertex/.style={
		draw,
		circle,
		inner sep=0.55em,
		thick
	},
	heavyvertex/.style={
		draw, 
		circle,
		inner sep=0.75em, 
		thick
	},
	nonheavyvertex/.style={
		draw, 
		circle,
		inner sep=0.4em, 
		thick
	},
	edge/.style={
		draw, 
		black, 
		solid,
		thick
	},
	mygray/.style={
		black!40
	},
	grayedge/.style={
		draw,
		mygray,
		solid,
		thick
	}]
	
	\drawVerticesWithXOffsetPrefix{-3}{};
	\drawVerticesWithXOffsetPrefix{17}{c};
	
	\path[->, thick, draw] (4.875em,  0.5em) -- (6.875em,  0.5em) {};
	\path[->, thick, draw] (6.875em, -0.5em) -- (4.875em, -0.5em) node [midway, label=below:{\small inverse}] {};
	
	\path[edge] (v1p1) -- (v9);
	\path[edge] (v1p2) -- (v7);
	\path[edge] (v1p3) -- (v11);
	\path[edge] (v2p1) -- (v10);
	\path[edge] (v2p2) -- (v8);
	\path[edge] (v2p3) -- (v12);
	\path[edge] (v3p1) -- (v13);
	\path[edge] (v3p2) -- (v15);
	\path[edge] (v4p1) -- (v14);
	\path[edge] (v4p2) -- (v16);
	\path[edge] (v5p1) -- (v19);
	\path[edge] (v5p2) -- (v17);
	\path[edge] (v6p1) -- (v20);
	\path[edge] (v6p2) -- (v18);
	
	\path[edge] (cv1p3) -- (cv3p1);
	\path[edge] (cv2p3) -- (cv4p1);
	\path[edge] (cv1p1) -- (cv5p2);
	\path[edge] (cv2p1) -- (cv6p2);
	\path[edge] (cv5p1) -- (cv6p1);
	\path[edge] (cv3p2) -- (cv4p2);
	\path[edge] (cv1p2) -- (cv7);
	\path[edge] (cv2p2) -- (cv8);
	
	\path[edge] (cv13) -- (cv15);
	\path[edge] (cv14) -- (cv16);
	\path[edge] (cv17) -- (cv19);
	\path[edge] (cv18) -- (cv20);
	\draw[edge] (cv9)  to [out=75, in=105, looseness=1.05]   (cv11);
	\draw[edge] (cv10) to [out=-75, in=-105, looseness=1.05] (cv12);
\end{tikzpicture}
 	}
	\caption{\centering The $t_b$-switching.}
	\label{fig:tbswitching}
 \end{subfigure}
 \vspace{1em}
\caption{\centering Switchings used in Phase 4.}
\end{figure}
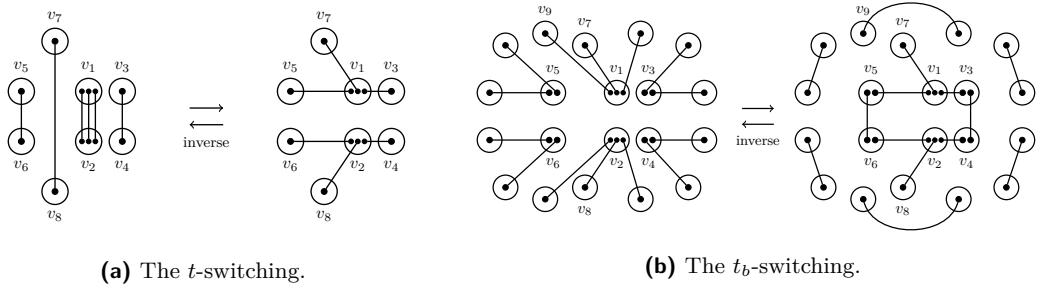

Phase 4 removes all triple-edges.
In each iteration, the algorithm first chooses a switching type $\tau$ from $\{t, t_a, t_b, t_c\}$, where type $\tau$ has probability~$\rho_{\tau}$.
The sum of these probabilities can be less than one, and the algorithm restarts with the remaining probability.
Overall, we have the following steps (where the  constants $\rho_\tau$, defined below, ensure uniformity -- cf.~\cite{DBLP:conf/focs/GaoW15}):

\begin{enumerate}
  \item Choose switching type $\tau$ with probability $\rho_{\tau}$, or restart with probability $1 - \sum_{\tau} \rho_{\tau}$.
 \item Pick a uniform random $\tau$-switching $S = (G, G')$.
If $\tau = t$, sample a uniform random triple-edge and three uniform random pairs, and switch them as shown in \autoref{fig:tswitching}.
If $\tau \neq t$, sample some uniform random $k$-stars, as exemplified for $t_b$ in \autoref{fig:tbswitching}, and switch them into the intended triplet.
 \item Restart (f-reject) if $S$ is not valid.
 The switching is valid if, for $\tau = t$, it removes the targeted triple-edge, or if, for $\tau \neq t$, it creates the intended triplet, without adding or removing (other) multi-edges or loops.
 \item Restart (b-reject) with probability $1 - \underline{b}_t(G';0) \underline{b}_t(G';1) / (b_t(G', \emptyset) b_t(G', v_1 v_3 v_5 v_7))$.
 \item Restart (b-reject) with probability $1 - \underline{b}_{\tau}(G') / b_{\tau}(G')$.
 \item Set $G \gets G'$.
\end{enumerate}

\noindent
For the b-rejection, incremental relaxation is used. In step $4$, there are two sub-rejections for the triplet, and in step $5$, there are sub-rejections for any additional pairs created (e.g. pairs that are not part of the triplet).
The $t$-switching creates no additional pairs, the $t_a$-switching creates three, the $t_b$-switching shown in \autoref{fig:tbswitching} creates six, and the $t_c$-switching nine.

The probability for step $4$ is computed as follows: first, set $b_t(G', \emptyset)$ to the number of simple ordered three-stars in $G'$.
Then, set $b_t(G', v_1 v_3 v_5 v_7)$ to the number of light simple ordered three-stars that a) do not share any nodes with the three-star $v_1 v_3 v_5 v_7$ created by~$S$, b) have no edge $v_1 v_2$ and no multi-edges $v_3 v_4$, $v_5 v_6$, $v_7 v_8$.
Let $B_k = \sum_{i=1}^{d_1} [d_{h+i}]_k$.
Then, the lower bounds are:
\begin{align}
 \underline{b}_t(G';0) &= M_3 - 18 m_t d_1^2 - 12 m_d d_1^2
 \\
 \underline{b}_t(G';1) &= L_3 - 18 m_t d_h^2 - 12 m_d d_h^2 - B_3 - 3 (m_t + m_d) B_2 - d_h^3 - 9 B_2
\end{align}

For step $5$: let $k$ be the number of additional pairs created by the switching.
Then, for pair $1 \leq i \leq k$, set $b_{\tau}(G', \overline{V}_{i+1}(S))$ to the number of simple ordered pairs in $G'$, that a) do not share nodes with the triplet or the previous $i - 1$ pairs, and b) have no edges that should have been removed by the switching (\eg in \autoref{fig:tbswitching}, $v_1 v_9$ cannot be an edge).
Finally, set $b_{\tau}(G') = \prod_{i=1}^k b_{\tau}(G', \overline{V}_{i+1}(S))$. The lower bound is:
\begin{align}
 \underline{b}_{\tau}(G') &= \prod\nolimits_{i=1}^{k} \underline{b}_\tau(G';i+1)
 \\
 \underline{b}_\tau(G';i+1) &= M_1 - 6 m_t - 4 m_d - 16 d_1 - 4 (i - 1) d_1 - 2 A_2
\end{align}

The type probabilities as computed as follows.
When initializing Phase~4, set $\rho_t = 1 - \varepsilon$ where $\varepsilon = 28 M_2^2 / M_1^3$, and set $\rho_{\tau} = 0$ for $\tau \in \{t_a, t_b, t_c\}$.
In each subsequent iteration, the probabilities are only updated after a $t$-switching $S = (G, G')$ is performed.
Then, first, let $i$ be the number of triple-edges in the graph $G'$, and let $i_1$ be the initial number of triple-edges after first entering Phase 4.
Now, define a parameter $x_i$:
\begin{align}
 &x_i = x_{i + 1} \rho_t \frac{\underline{b}_t(G';0) \underline{b}_t(G';1)}{\overline{f}_t(i + 1)} + 1,
\end{align}
where $x_{i_1} = 1$ and $\overline{f}_t(i) = 12 i M_1^3$.

Define $\overline{f}_{t_a} = 3 M_3 L_3 M_2^2$, $\overline{f}_{t_b} = 3 M_3 L_3 M_2^4$, and $\overline{f}_{t_c} = M_3 L_3 M_2^6$.
Then, update the probability $\rho_{\tau}$ for $\tau \in \{t_a, t_b, t_c\}$ as follows:
\begin{align}
 &\rho_{\tau} = \frac{x_{i + 1}}{x_i} \rho_t \frac{\overline{f}_{\tau}}{\underline{b}_{\tau}(G') \overline{f}_t(i + 1)}
\end{align}
Finally, the algorithm enters Phase 5.

\subsection{Phase 5: removal of light double-edges}
\label{subsec:phase5}
Similar to Phase 4, Phase 5 uses multiple different switchings.
The $d$-switching (see \autoref{fig:typeIswitching}) removes double-edges. The booster switchings create so-called \emph{doublets} consisting of a simple two-star $v_1 v_3 v_5$ and a light simple two-star $v_2 v_4 v_6$ that do not share any nodes.
Let $m_1$, $m_2$, and $m_3$ denote the multiplicities of the edges $v_1 v_2$, $v_3 v_4$ and $v_5 v_6$ in a doublet, respectively.
Then, the $d$-switching creates the doublet with $\max(m_1, m_2, m_3) = 0$. The booster switchings create all other doublets where $\max(m_1, m_2, m_3) \leq 2$, \ie where some of the edges are single-edges or double-edges.
We identify each booster switching by the doublet created, \eg type~$\tau = (1,2,0)$ shown in \autoref{fig:typeVbswitching} creates a doublet where $m_1 = 1$, $m_2 = 2$, and $m_3 = 0$.

\begin{figure}[t]
 \begin{subfigure}{.45\textwidth}
 	\resizebox{\textwidth}{!}{
			\def\node#1{}
\newcommand{\drawVerticesWithXOffsetPrefix}[2]{
	\def\y{1.5em};
	\def\x{3em}
	\def\px{.5em}
	\def\ld{-.25em}
	\def\a{20}
	\node[point] (#2i1) at (#1em - .5*\px, \y) {};
	\node[point] (#2i2) at (#1em + .5*\px, \y) {};
	\node[2pointvertex, label={[label distance=\ld]\a:$v_1$}] (#2i) at (#1em, \y) {};
	\node[point] (#2j1) at (#1em - .5*\px, -\y) {};
	\node[point] (#2j2) at (#1em + .5*\px, -\y) {};
	\node[2pointvertex, label={[label distance=\ld]-\a:$v_2$}] (#2j) at (#1em, -\y) {};
	\coordinate (#2v3coords) at (#1em - \x,  \y);
	\coordinate (#2v4coords) at (#1em - \x, -\y);
	\node[point] (#2v3) at (#2v3coords) {};
	\node[point] (#2v4) at (#2v4coords) {};
	\node[2pointvertex, label={[label distance=\ld]\a:$v_3$}]  at (#2v3coords) {};
	\node[2pointvertex, label={[label distance=\ld]-\a:$v_4$}] at (#2v4coords) {};
	\coordinate (#2v5coords) at (#1em + \x,  \y);
	\coordinate (#2v6coords) at (#1em + \x, -\y);
	\node[point] (#2v5) at (#2v5coords) {};
	\node[point] (#2v6) at (#2v6coords) {};
	\node[2pointvertex, label={[label distance=\ld]\a:$v_5$}]  at (#2v5coords) {};
	\node[2pointvertex, label={[label distance=\ld]-\a:$v_6$}] at (#2v6coords) {};
}
\begin{tikzpicture}[
	point/.style={
		draw,
		circle,
		inner sep=0,
		fill,
		minimum width=0.35em,
		minimum height=0.35em
	},
	heavyvertex/.style={
		draw, 
		circle,
		inner sep=0.75em, 
		thick
	},
	nonheavyvertex/.style={
		draw, 
		circle,
		inner sep=0.4em, 
		thick
	},
	2pointvertex/.style={
		draw,
		circle,
		inner sep=0.55em,
		thick
	},
	edge/.style={
		draw, 
		black, 
		solid,
		thick
	}]
	
	\drawVerticesWithXOffsetPrefix{-3}{};
	\drawVerticesWithXOffsetPrefix{9}{c};
	
	\path[->, thick, draw] (2em,  0.5em) -- (4em,  0.5em) {};
	\path[->, thick, draw] (4em, -0.5em) -- (2em, -0.5em) node [midway, label=below:{\small inverse}] {};
	
	\path[edge] (i1) -- (j1);
	\path[edge] (i2) -- (j2);
	\path[edge] (v3) -- (v4);
	\path[edge] (v5) -- (v6);
	
	\path[edge] (ci1) -- (cv3);
	\path[edge] (ci2) -- (cv5);
	\path[edge] (cj1) -- (cv4);
	\path[edge] (cj2) -- (cv6);
	
	\coordinate (cv13) at (-3em,     4.5em);
	\coordinate (cv9)  at (4.242em,  4.5em); 
	\coordinate (cv14) at (-3em,    -4.5em);
	\coordinate (cv10) at (4.242em, -4.5em);
\end{tikzpicture}
	}
		\caption{\centering The $d$-switching.}
		\label{fig:typeIswitching}
 \end{subfigure}
\hfill
 \begin{subfigure}{.5\textwidth}
 	\resizebox{\textwidth}{!}{
			\def\node#1{}
\newcommand{\drawVerticesWithXOffsetPrefix}[2]{
	\def\y{1.5em}
	\def\x{3em}
	\def\sqx{2.121em}
	\def\spx{.4em}
	\def\px{.5em}
	\def\a{10}
	\def\ld{-.25em}
	\coordinate (#2v1coords)      at (#1em,  \y);
	\node[smallpoint] (#2v1p1)    at (#1em - \spx,  \y) {};
	\node[smallpoint] (#2v1p2)    at (#1em,         \y) {};
	\node[smallpoint] (#2v1p3)    at (#1em + \spx,  \y) {};
	\node[2pointvertex, label={[label distance=\ld]\a:$v_1$}] at (#2v1coords) {};
	\coordinate (#2v5coords) at (#1em + \x,  \y);
	\node[point] (#2v5) at (#2v5coords) {};
	\node[2pointvertex, label={[label distance=\ld]\a:$v_5$}] at (#2v5coords) {};
	\coordinate (#2v7coords) at (#1em,  3*\y);
	\node[point] (#2v7) at (#2v7coords) {};
	\node[2pointvertex, label={[label distance=\ld]100:$v_7$}] at (#2v7coords) {};
	\coordinate (#2v9coords) at (#1em + \sqx, 3*\y);
	\node[point] (#2v9) at (#2v9coords) {};
	\node[2pointvertex] at (#2v9coords) {}; 

	\coordinate (#2v2coords)      at (#1em, -\y);
	\node[smallpoint] (#2v2p1)    at (#1em - \spx, -\y) {};
	\node[smallpoint] (#2v2p2)    at (#1em       , -\y) {};
	\node[smallpoint] (#2v2p3)    at (#1em + \spx, -\y) {};
	\node[2pointvertex, label={[label distance=\ld]-\a:$v_2$}] at (#2v2coords) {};
	\coordinate (#2v6coords) at (#1em + \x, -\y);
	\node[point] (#2v6) at (#2v6coords) {};
	\node[2pointvertex, label={[label distance=\ld]-\a:$v_6$}] at (#2v6coords) {};
	\coordinate (#2v8coords) at (#1em, -3*\y);
	\node[point] (#2v8) at (#2v8coords) {};
	\node[2pointvertex] at (#2v8coords) {};
	\coordinate (#2v10coords) at (#1em + \sqx, -3*\y);
	\node[point] (#2v10) at (#2v10coords) {};
	\node[2pointvertex] at (#2v10coords) {}; 
	
	\node[smallpoint] (#2v3p1) at (#1em - \x - \spx, \y) {};
	\node[smallpoint] (#2v3p2) at (#1em - \x,        \y) {};
	\node[smallpoint] (#2v3p3) at (#1em - \x + \spx, \y) {};
	\node[2pointvertex, label={[label distance=\ld]\a:$v_3$}] (#2v3) at (#1em - \x,  \y) {};
	\coordinate (#2v11coords) at (#1em - \x, 3*\y) {};
	\node[point] (#2v11) at (#2v11coords) {};
	\node[2pointvertex]  at (#2v11coords) {};
	\coordinate (#2v13coords) at (#1em - \x - \sqx, 3*\y) {};
	\node[point] (#2v13) at (#2v13coords) {};
	\node[2pointvertex]  at (#2v13coords) {};
	\coordinate (#2v15coords) at (#1em - 2*\x, \y) {};
	\node[point] (#2v15) at (#2v15coords) {};
	\node[2pointvertex]  at (#2v15coords) {};
	
	\node[smallpoint] (#2v4p1) at (#1em - \x - \spx, -\y) {};
	\node[smallpoint] (#2v4p2) at (#1em - \x,        -\y) {};
	\node[smallpoint] (#2v4p3) at (#1em - \x + \spx, -\y) {};
	\node[2pointvertex, label={[label distance=\ld]-\a:$v_4$}] (#2v4) at (#1em - \x, -\y) {};
	\coordinate (#2v12coords) at (#1em - \x, -3*\y) {};
	\node[point] (#2v12) at (#2v12coords) {};
	\node[2pointvertex]  at (#2v12coords) {};
	\coordinate (#2v14coords) at (#1em - \x - \sqx, -3*\y) {};
	\node[point] (#2v14) at (#2v14coords) {};
	\node[2pointvertex]  at (#2v14coords) {};
	\coordinate (#2v16coords) at (#1em - 2*\x, -\y) {};
	\node[point] (#2v16) at (#2v16coords) {};
	\node[2pointvertex]  at (#2v16coords) {};
}
\begin{tikzpicture}[
	point/.style={
		draw,
		circle,
		inner sep=0,
		fill,
		minimum width=0.35em,
		minimum height=0.35em
	},
	smallpoint/.style={
		draw,
		circle,
		inner sep=0,
		fill,
		minimum width=0.28em,
		minimum height=0.28em,
	},
	heavyvertex/.style={
		draw, 
		circle,
		inner sep=0.75em, 
		thick
	},
	nonheavyvertex/.style={
		draw, 
		circle,
		inner sep=0.4em, 
		thick
	},
	2pointvertex/.style={
		draw,
		circle,
		inner sep=0.55em,
		thick
	},
	edge/.style={
		draw, 
		black, 
		solid,
		thick
	}]
	
	\drawVerticesWithXOffsetPrefix{-3}{};
	\drawVerticesWithXOffsetPrefix{12}{c};
	
	\path[->, thick, draw] (2em,  0.5em) -- (4em,  0.5em) {};
	\path[->, thick, draw] (4em, -0.5em) -- (2em, -0.5em) node [midway, label=below:{\small inverse}] {};
	
	\path[edge] (v1p1) -- (v7);
	\path[edge] (v1p2) -- (v9);
	\path[edge] (v1p3) -- (v5);
	\path[edge] (v2p1) -- (v8);
	\path[edge] (v2p2) -- (v10);
	\path[edge] (v2p3) -- (v6);
	\path[edge] (v3p1) -- (v15);
	\path[edge] (v3p2) -- (v13);
	\path[edge] (v3p3) -- (v11);
	\path[edge] (v4p1) -- (v16);
	\path[edge] (v4p2) -- (v14);
	\path[edge] (v4p3) -- (v12);
	
	\path[edge] (cv3p3) -- (cv1p1);
	\path[edge] (cv4p3) -- (cv2p1);
	\path[edge] (cv3p2) -- (cv4p2);
	\path[edge] (cv3p1) -- (cv4p1);
	\path[edge] (cv1p2) -- (cv2p2);
	\path[edge] (cv1p3) -- (cv5);
	\path[edge] (cv2p3) -- (cv6);
	\path[edge] (cv11)  -- (cv7);
	\path[edge] (cv12)  -- (cv8);
	\path[edge] (cv15)  -- (cv16);
	\draw[edge] (cv13) to [out=75, in=105, looseness=1]   (cv9);
	\draw[edge] (cv14) to [out=-75, in=-105, looseness=1] (cv10);
\end{tikzpicture}
 	}
		\caption{\centering The type $\tau = (1,2,0)$ booster switching.}
		\label{fig:typeVbswitching}
 \end{subfigure}
 \vspace{1em}
 \caption{\centering Switchings used in Phase 5.}
\end{figure}

Phase 5 repeats the following steps until all double-edges are removed:

\begin{enumerate}
  \item Choose switching type $\tau$ with probability $\rho_{\tau}$ or restart with the probability $1 - \sum_{\tau} \rho_{\tau}$.
 \item Pick a uniform random $\tau$-switching $S = (G, G')$. If $\tau = d$, sample a uniform random double-edge and two uniform random pairs, and switch them as shown in \autoref{fig:typeIswitching}.
If $\tau \neq d$, sample a number of uniform random $k$-stars, as exemplified in \autoref{fig:typeVbswitching}, and switch them into the intended doublet and a number of additional simple edges.
 \item Restart (f-reject) if $S$ is not valid.
The switching is valid if, for $\tau = d$, it removes the targeted double-edge or, for $\tau \neq d$, it creates the intended doublet without adding or removing (other) multi-edges or loops.
 \item Restart (b-reject) with probability $1 - \underline{b}_d(G';0) \underline{b}_d(G';1) / (b_d(G', \emptyset) b_d(G', v_1 v_3 v_5))$.
 \item Restart (b-reject) with probability $1 - \underline{b}_{\tau}(G') / b_{\tau}(G')$.
 \item Set $G \gets G'$.
\end{enumerate}

\noindent
For the b-rejection, incremental relaxation is used. Step $4$ contains the sub-rejections for the doublet and step $5$ the rejections for any additional pairs created by the switching.
In general, the number of additional pairs created by type $\tau = (m_1, m_2, m_3)$ is $k = I_{m_1 \geq 1} m_1 + I_{m_2 \geq 1} (m_2 + 2) + I_{m_3 \geq 1} (m_3 + 2)$ where $I$ denotes the indicator function.

For step $4$: first, set $b_d(G', \emptyset)$ to the number of simple ordered two-stars in $G'$.
Then, set $b_d(G', v_1 v_3 v_5)$ to the number of light simple ordered two-stars in $G'$ that do not share any nodes with the two-star $v_1 v_3 v_5$ created by $S$.
The lower bounds are:
\begin{align}
 \underline{b}_d(G';0) &= M_2 - 8 m_d d_1
 \\
 \underline{b}_d(G';1) &= L_2 - 8 m_d d_h - 6 B_1 - 3 d_h^2
\end{align}

For step $5$: let $k$ be the number of additional pairs created by the switching.
Then, for pair $1 \leq i \leq k$, set $b_{\tau}(G', \overline{V}_{i+1}(S))$ to the number of simple ordered pairs in $G'$, that a) do not share nodes with the doublet or the previous $i - 1$ pairs, and b) have no edges that should have been removed by the switching (\eg in \autoref{fig:typeVbswitching}, $v_1 v_7$ cannot be an edge).
Finally, set $b_{\tau}(G') = \prod_{i=1}^k b_{\tau}(G', \overline{V}_{i+1}(S))$.
The lower bound is:
\begin{align}
\underline{b}_{\tau}(G') &= \prod\nolimits_{i=1}^{k} \underline{b}_\tau(G';i+1)
 \\
\underline{b}_\tau(G';i+1) &= M_1 - 4 m_d - 12 d_1 - 4 (i - 1) d_1 - 2 A_2
\end{align}

The type probabilities are computed as follows.
When initializing Phase~5, set $\rho_d = 1 - \xi$ where
\begin{align}
 \xi = \frac{32 M_2^2}{M_1^3} + \frac{36 M_4 L_4}{M_2 L_2 M_1^2} + \frac{32 M_3^2}{M_1^4},
\end{align}
and set $\rho_{\tau} = 0$ for all types $\tau \neq d$.
The probabilities are updated after a switching $S = (G, G')$ is performed that changes the number of double-edges.
Then, first, let $i$ denote the new number of double-edges in $G'$, let $i_1$ denote the initial number of double-edges after first entering Phase 5, and let $\rho_{d}(i)$ denote the probability of type $d$ on a graph with $i$ double-edges.
Now, define a parameter $x_i$:
\begin{align}
 x_i = x_{i + 1} \rho_{d}(i+1) \frac{\underline{b}_d(G';0) \underline{b}_d(G';1)}{\overline{f}_d(i + 1)} + 1,
\end{align}
where $x_{i_1} = 1$ and $\overline{f}_d(i) = 4 i M_1^2$.

Then, to update the probability of a booster switching type $\tau$, there are two cases: (1) if $\tau = (m_1, m_2, m_3)$ adds double-edges (\ie $\max(m_1, m_2, m_3) = 2$) and the number of double-edges if a switching of this type was performed would be higher than $i_1 - 1$, set $\rho_{\tau} = 0$.
Otherwise, (2) let $i'$ denote the new number of double-edges if a switching of this type was performed, and define
\begin{align}
 \overline{f}_{\tau} &= M_{k_1} L_{k_1} (I_{k_2 \geq 2} M_{k_2}^2 + I_{k_2 \leq 1} 1) (I_{k_3 \geq 2} M_{k_3}^2 + I_{k_3 \leq 1} 1),
\end{align}
where $k_1 = m_1 + 2$, $k_2 = m_2 + 1$, $k_3 = m_3 + 1$.

Then set:
\begin{align}
 \rho_{\tau} &= \frac{x_{i' + 1}}{x_i} \rho_{d_{i' + 1}} \frac{\overline{f}_{\tau}}{\underline{b}_{\tau}(i') \overline{f}_d(i' + 1)}
 \\
 \rho_d &= 1 - \rho_{1,0,0} - \xi.
\end{align}

If Phase 5 terminates, all non-simple edges are removed, and the final graph $G$ is output.

\section{Adjustments to the algorithm}
\label{sec:new-switchings}
\label{sec:proofs}
In this section, we describe our additions and adjustments to the \ipwl algorithm sketched in \cite{DBLP:conf/focs/ArmanGW19}.

\subsection{New switchings in Phase 4}
Phase 4 of \pld only uses the $t$-switching \cite{DBLP:conf/soda/GaoW18}.
There, the rejection probability is small enough so that no booster switchings are needed.
The overall running-time of Phase 4 in \pld however, is superlinear, as computing the probability of a b-rejection requires counting the number of valid $t$-switchings that produce a graph.
By using incremental relaxation \cite{DBLP:conf/focs/ArmanGW19}, the cost of computing the b-rejection probability becomes sublinear, as it only requires us to count simpler structures in the graph.
However, when applying incremental relaxation to Phase 4, the probability of a b-rejection increases, as the lower bounds on the number of those structures are less tight, and the overall running-time remains superlinear.

We address this issue by using booster switchings in Phase 4 to reduce the b-rejection probability (analogous to Phase 5 of \pld).
To this end, we add three new switchings: the $t_a$, $t_b$ and $t_c$ switching (see \autoref{subsec:phase4}).
This is done entirely analogous to Phase 5 of \pld, which also uses multiple switchings in the same Phase.
We first derive an equation for the expected number of times that a graph is produced by a type $\tau \in \{t, t_a, t_b, t_c\}$ switching.
Then we set equal the expected number of times for each graph in the same set $\mathcal{S}(\rrr)$.
The resulting system of equations is fully determined by choosing an upper bound $\varepsilon$ on the probability of choosing a type $\tau \neq t$.
We can then derive the correct probabilities $\rho_\tau$ for each type as a function of $\rrr$.

\begin{lemma}
\label{th:ph4b}
 Let $\degseq$ be a plib sequence with exponent $\gamma > 21/10 + \sqrt{61}/10 \approx 2.88102$.
 Then, given $\degseq$ as input, the probability of a b-rejection in Phase 4 of \ipwl is $o(1)$.
\end{lemma}
 
\begin{proof}
We first show that the number of iterations in Phase 4 is at most $\Oh{L_3 M_3 / M_1^3}$.
First, recall that the algorithm only enters Phase 4 if the graph satisfies the Phase 3, 4 and 5 preconditions.
In particular, the graph may contain at most $L_3 M_3 / M_1^3$ triple-edges.
Phase 4 terminates once all triple-edges are removed.
In each iteration, a triple-edge is removed if we choose a type $t$-switching.
The probability of choosing a $t$-switching is $\rho_t = 1 - \varepsilon$, where $\varepsilon = 28 M_2^2 / M_1^3$, and it can be verified that the probability of choosing any other switching is bounded by $\varepsilon$.
In addition, we know that $M_k = \Oh{n^{k/\gamma - 1}}$ for $k \geq 2$ and $M_1 = \Theta(n)$ \cite{DBLP:conf/soda/GaoW18}, and we have $\varepsilon = \Oh{n^{2/(\gamma - 1)} / n^3} = o(1)$ assuming that $\gamma > 21/10 + \sqrt{61}/10 \approx 2.88102$.
Thus, a triple-edge is removed in each iteration with probability $1 - o(1)$, and as the graph may contain at most $L_3 M_3 / M_1^3$ triple-edges, the total number of iterations is at most $\Oh{L_3 M_3 / M_1^3}$.

Now, we show that the probability of a b-rejection vanishes with $n$.
First, it is easy to verify that the probability of a b-rejection in step $4$ dominates the probability of a rejection in step $5$ (compare \autoref{subsec:phase4}).
The probability of a b-rejection in step $4$ is $1 - \underline{b}_t(G';0) \underline{b}_t(G';1) / (b_t(G', \emptyset)  b_t(G', v_1 v_3 v_5 v_7))$. It can be shown that $\underline{b}_t(G';0) = \Omega(M_3)$, $\underline{b}_t(G';1) = \Omega(L_3 - B_3)$, and $b_t(G', \emptyset) b_t(G', v_1 v_3 v_5 v_7) \leq M_3 L_3$.
Thus, the probability of a b-rejection is at most $\Oh{M_3 B_3 / M_3 L_3}$.
In addition, as shown above, the number of iterations of Phase 4 is at most $\Oh{L_3 M_3 / M_1^3}$.
Then, the overall probability of a b-rejection during all of Phase 4 is at most $\Oh{M_3 B_3 / M_1^3} = \Oh{n^{3/(\gamma - 1)} n^{1 - \delta (\gamma - 4)} / n^3} = o(1)$ for $\gamma > 21/10 + \sqrt{61}/10 \approx 2.88102$.
\end{proof}

\subsection{New switchings in Phase 5}
Phase 5 of \pld uses the type-I switching (this is the same as the $d$-switching in \ipwl), as well as a total of six booster switchings called type-III, type-IV, type-V, type-VI and type-VII.
These booster switchings create the doublets where each of the "bad edges" can either be a non-edge or single-edge, i.e. $\max\{m_1, m_2, m_3\} = 1$.
For Phase 5 of \pld, this suffices to ensure that the b-rejection probability is small enough.
However, similar to Phase 4, applying incremental relaxation to reduce the computational cost increases the rejection probability, leading to a superlinear running-time overall.

To further reduce the probability of a b-rejection, we add booster switchings that create the doublets where one or more of the "bad edges" is a double-edge, i.e. $\max\{m_1, m_2, m_3\} = 2$ (see \autoref{subsec:phase5}).
We then integrate the new switchings to Phase 5 of \ipwl by deriving the correct probabilities $\rho_\tau$ and increasing the constant $\xi$ used to bound the probabilities of the types $\tau \notin \{d, (1, 0, 0)\}$.

\begin{lemma}
\label{th:ph5b}
 Let $\degseq$ be a plib sequence with exponent $\gamma > 21/10 + \sqrt{61}/10 \approx 2.88102$.
 Then, given $\degseq$ as input, the probability of a b-rejection in Phase 5 of \ipwl is $o(1)$.
\end{lemma}
 
\begin{proof}
Analogous to the proof of \autoref{th:ph4b}, we first bound the number of iterations in Phase 5.
In each iteration, a double-edge is removed if we chose the $d$-switching, and a $d$-switching is chosen with probability $1 - \rho_{1,0,0} - \xi$, where $\xi = 32 M_2^2 / M_1^3 + 36 M_4 L_4 / M_2 L_2 M_1^2 + 32 M_3^2 / M_1^4$.
It can be verified that the probability of choosing any of the other switchings is bounded by $\xi$.
In Phase 5 of \pld, the probability of choosing a type-I switching is set to $1 - \rho_{III} - \xi'$, where $\xi' = 32 M_2^2 / M_1^3$.
As the probability of not choosing a type-I switching in \pld vanishes with $n$, we know that the terms $\rho_{1,0,0} = \rho_{III}$ and $32 M_2^2 / M_1^3$ vanish with $n$.
For the remaining two terms, first note that $L_{k+1} \leq L_{k} d_h = \Oh{L_k n^{\delta}}$ and $M_{k+1} \leq M_{k} d_1 = \Oh{M_k n^{1/(\gamma - 1)}}$ for $k~\geq~2$  \cite{DBLP:conf/soda/GaoW18}.
Then, we have $36 M_4 L_4 / M_2 L_2 M_1^2 \leq M_2 d_1^2 L_2 d_h^2 / M_2 L_2 M_1^2 = \Oh{n^{2/(\gamma - 1) + 2 \delta} / n^2} = o(1)$, and $32 M_3^2 / M_1^4 = \Oh{n^{6/(\gamma - 1)} / n^4} = o(1)$ assuming $\gamma > 21/10 + \sqrt{61}/10 \approx 2.88102$.
Thus, a double-edge is removed in each iteration with probability $1 - o(1)$, and as a graph satisfying the Phase 3, 4 and 5 preconditions may contain at most $L_2 M_2 / M_1^2$ double-edges, the total number of iterations in Phase 5 is at most $\Oh{L_2 M_2 / M_1^2}$.

We now show that the probability of a b-rejection is small enough.
Again, the probability of a b-rejection in step $4$ dominates the probability of a b-rejection in step $5$ (compare \autoref{subsec:phase5}).
The rejection probability in step $4$ is $1 - \underline{b}_d(G';0) \underline{b}_d(G';1) / (b_d(G', \emptyset) b_d(G', v_1 v_3 v_5))$.
Note that $\underline{b}_d(G';0) = \Omega(M_2), \underline{b}_d(G';1) = \Omega(L_2 - A_2)$, and $b_d(G', \emptyset) b_d(G', v_1 v_3 v_5) \leq M_2 L_2$.
Then, the overall probability of a b-rejection in Phase 5 is at most $\Oh{M_2 A_2 / M_1^2} = \Oh{n^{2/(\gamma - 1)} n^{(2 \gamma - 3)/(\gamma - 1)^2} / n^2} = o(1)$ for $\gamma > 21/10 + \sqrt{61}/10 \approx 2.88102$.
\end{proof}

\subsection{Expected running-time}

We now use \autoref{th:ph4b} and \autoref{th:ph5b} to bound the expected running-time of the algorithm.
\begin{theorem}
\label{th:ipwlrt}
 Let $\degseq$ be a plib sequence with exponent $\gamma > 21/10 + \sqrt{61}/10 \approx 2.88102$.
 Then, given $\degseq$ as input, the expected running-time of \ipwl is $O(n)$.
\end{theorem}

\begin{proof}
 From \cite{DBLP:conf/focs/ArmanGW19}, we know that the running-time of each individual phase (e.g. computation of rejection parameters, etc.) is at most $O(n)$, and in addition, we know that the probability of an f-rejection in any Phase is $o(1)$ and the probability of a b-rejection in Phases 1, 2 or 3 is $o(1)$.
 By \autoref{th:ph4b} and \autoref{th:ph5b}, the probability of a b-rejection in Phase 4 or 5 is $o(1)$.
 Therefore, the expected number of restarts is $O(1)$, and the overall running-time of \ipwl is $O(n)$.
\end{proof}

\section{Implementation}
\label{sec:implementation}
In this section we highlight some aspects of our \ipwl implementation.
The generator is implemented in modern \textsc{C++} and relies on Boost Multiprecision\footnote{%
	\url{https://github.com/boostorg/multiprecision} (V~1.76.0)
}
to handle large integer and rational numbers that occur even for relatively small inputs.

\subsection{Graph representation}
\ipwl requires a dynamic graph representation capable of adding and removing edges, answering edge existence and edge multiplicity queries, enumerating a node's neighborhood, and sampling edges weighted by their multiplicity.
A careful combination of an adjacency vector and a hash map yields expected constant work for all operations.

In practice, however, we find that building and maintaining these structures is more expensive than using a simpler, asymptotically sub-optimal, approach.
To this end, we exploit the small and asymptotically constant average degree of plib degree sequences and the fact that most queries do not modify the data structure.
This allows us to only use a compressed sparse row (CSR) representation that places all neighborhoods in a contiguous array~$A_C$ and keep the start indices~$A_I$ of each neighborhood in a separate array; neighborhoods are maintained in sorted order and neighbors may appear multiple times to encode multi-edges.

Given an edge list, we can construct a CSR in time $\Oh{n+m}$ using integer sorting.
A subsequent insertion or deletion of edge $uv$ requires time $\Oh{\deg(u) + \deg(v)}$; these operations, however, occur at a low rate.
Edge existence and edge multiplicity queries for edge $uv$ are possible in time $\Oh{\log \min(\deg(u), \deg(v))}$ by considering the node with the smaller neighborhood (as $\degseq$ is ordered $u \le v$ implies $\deg(u) \ge \deg(v)$).
Assuming plib degrees, these operations require constant expected work.
Randomly drawing an edge weighted by its multiplicity is implemented by drawing a uniform index $j$ for $A_C$ and searching its incident node in $A_I$ in time $\Oh{\log n}$\footnote{Observe that constant time look-ups are straight-forward by augmenting each entry in $A_C$ with the neighbor. We, however, found the contribution of the binary search non-substantial.}.

Several auxiliary structures (\eg indices to non-simple edges, numbers of several sub-structures, et cetera) are maintained requiring $\Oh{m}$ work in total. 
Where possible, we delay their construction to the beginning of Phase~3 in order to not waste work in case of an early rejection in Phases~1 or 2.

\subsection{Parallelism}\label{subsec:impl-parallel}
While \ipwl seems inherently sequential (e.g. due to the dependence of each switching step on the previous steps), it is possible to parallelize aspects of the algorithm.
In the following we sketch two non-exclusive strategies.
These approaches are practically significant, but are not designed to yield a sub-linear time complexity.

	\paragraph*{Intra-run}
	As we discuss in \autoref{sec:evaluation}, the implementation's runtime is dominated by the sampling of the initial multigraph and construction of the CSR.
	These in turn spend most time with random shuffling and sorting. Both can be parallelized~\cite{DBLP:journals/corr/abs-2009-13569,DBLP:journals/ipl/Sanders98}.

	\paragraph*{Inter-run} 
		If \ipwl restarts, the following attempt is independent of the rejected one.
		Thus, we can execute $P$ instances of \ipwl in parallel and return the ``first'' accepted result.
		Synchronization is only used to avoid a selection bias towards quicker runs:
		all processors assign globally unique indices to their runs and update them after each restart.
		We return the accepted result with smallest index and terminate processes working on results with larger indices prematurely.
		The resulting speed-up is bounded by the number of restarts which is typically a small constant.

\subsection{Configuration model}
As the majority of time is spend sampling the initial graph~$G$ and building its CSR representation, we carefully optimize this part of our implementation.
First, we give an extended description of the configuration model that remains functionally equivalent to \autoref{sec:algorithm}.

\textsl{
Given a degree-sequence $\degseq = (d_1, \ldots, d_n)$, let $G$ be a graph with $n$ nodes and no edges.
For each node $u \in V$ place $d_u$ marbles labeled $u$ into an urn.
Then, randomly draw without replacement two marbles with labels $a$ and $b$, respectively.
Append label~$a$ to an initially empty sequence $A$ and analogously label~$b$ to $B$.
Finally, add for each $1 \le i \le m$ the edge $\{A[i], B[i]\}$ to $G$.
}

We adopt a common strategy~\cite{DBLP:journals/corr/abs-2003-00736} to implement sampling without replacement.
First produce a sequence $N[1\dots2m]$ representing the urn, \ie the value $i$ is contained $d_i$ times.
Then randomly shuffle $N$ and call the result $N'$.
Finally, partition $N'$ arbitrarily to obtain the aforementioned sequences $A$ and $B$ of equal sizes.
For our purpose, it is convenient to choose the first and second halves of $N'$, \ie $A = N'[1 \ldots m]$ and $B = N'[m{+}1\ldots2m]$.
\footnote{To ``shuffle'' or ``random permute'' refers to the process of randomly reordering a sequence such that any permutation occurs with equal probability.}

Our parallel implementation shuffles~$N$ with a shared memory implementation based on~\cite{DBLP:journals/ipl/Sanders98}.
We then construct a list of all pairs in both orientations and sort it lexicographically in parallel~\cite{DBLP:journals/corr/abs-2009-13569}.
In the resulting sequence, each neighborhood is a contiguous subsequence.
Hence, we can assign the parallel workers to independent subproblems by aligning them to the neighborhood boundaries.
The sequential algorithm follows the same framework to improve locality of reference in the data accesses.
It uses a highly tuned Fisher-Yates shuffle based on~\cite{DBLP:journals/tomacs/Lemire19} and the integer sorting SkaSort\footnote{\url{https://github.com/skarupke/ska_sort}}.

Both shuffling algorithms are modified almost halving their work.
The key insight is that the distribution of graphs sampled remains unchanged if we only shuffle $A$ and allow an arbitrary permutation of $B$ (or vice versa).
This can be seen as follows.
Assume we sampled $A$ and $B$ as before and computed graph $G_{A,B}$.
Then, we let an adversary choose an arbitrary permutation~$\pi_B$ of $B$ without knowing~$A$.
If we apply $\pi_B$ to $B$ before adding the edges, the resulting $G_{A,\pi_B(B)}$ is in general different from $G_{AB}$.
We claim, however, that $G_{A,B}$ and $G_{A,\pi_B(B)}$ both are equally distributed samples of the configuration model.
We can recover the original graph by also applying $\pi_B$ to $A$, \ie $G_{\pi_B(A),\pi_B(B)} = G_{A,B}$.
Let $P_m$ denote the set of all $m!$ permutations of a sequence of length~$m$, and note that the composition $\circ\colon P_m \times P_m \to P_m$ is a bijection.
Further recall that $A$ is randomly shuffled and all its permutations $\pi_A \in P_m$ occur with equal probability.
Thus, as $\pi_A$ is uniformly drawn from $P_m$, so is $(\pi_A \circ \pi_B) \in P_m$.
In conclusion, the distribution of edges is independent of the adversary's choice.

To exploit this observation, we partition $N$ into two subsequences $N'[1 \ldots k]$ and $N'[k+1 \ldots 2m]$.
Each element is assigned to one subsequence using an independent and fair coin flip.
While partitioning and shuffling are both linear time tasks, in practice, the former can be solved significantly faster (in the parallel algorithm~\cite{DBLP:journals/ipl/Sanders98}, it is even a by-product of the assignment of subproblems to workers).
Observe that with high probability both sequences have roughly equal size, \ie $|k - m| = \Oh{\sqrt m}$.
We then only shuffle the larger one (arbitrary tie-breaking if $k=m$), and finally output the pairs $(N'[i], N'[m+i])$ for all $1 \le i \le m$.
\vspace{0.299em}

\section{Empirical evaluation}
\label{sec:evaluation}
In the following, we empirically investigate our implementation of \ipwl.

To reaffirm the correctness of our implementation and empirically support the uniformity of the sampled graphs, we used unit tests and statistical tests.
For instance, we carried out $\chi^2$-tests over the distribution of all possible graphs for dedicated small degree sequences $\degseq$ where it is feasible to fully enumerate $\mathcal G(\degseq)$.
Additionally, we assert the plausibility of rejection parameters.

The widely accepted, yet approximate, Edge-Switching MCMC algorithm provides a reference to existing solutions.
We consider two implementations:
\textsc{NetworKit-ES}, included in NetworKit~\cite{DBLP:journals/netsci/StaudtSM16} and based on~\cite{DBLP:conf/alenex/GkantsidisMMZ03}, was selected for its readily availability and flexibility.
\textsc{Fast-ES} is our own solution that is at least as fast as all open sequential implementations we are aware of.
For the latter, we even exclude the set-up time for the graph data structure.
To their advantage, we execute an average of $10$ switches per edge (in practice, common  choices~\cite{Milo2003,DBLP:conf/alenex/GkantsidisMMZ03,DBLP:conf/waw/RayPS12} are $10$ to $30$).
Increasing this number improves the approximation of a uniform distribution, but linearly increases the work.

In each experiment below, we generate between $100$ and $1000$ random power-law degree sequences with fixed parameters $n$, $\gamma$, and minimal degree $d_{min}$ analogously to the \textsc{PowerlawDegreeSequence} generator of NetworKit.
Then, for each sequence, we benchmark the time it takes for the implementations to generate a graph and report their average.
In the plots, a shaded area indicates the 95\%-confidence interval.
The benchmarks are built with GNU g++-9.3 and executed on a machine equipped with an AMD EPYC 7452 (32 cores) processor and 128~GB RAM running Ubuntu 20.04.

\paragraph*{Running-time scaling in n}
In \autoref{fig:dep_n} we report the performance of \ipwl and the {Edge-Switching} implementations for degree sequences with $\gamma \approx 2.88$, $d_{min} = 1$, and $n = 2^k$ for integer values $10 \le k \le 28$.
Our \ipwl implementation generates a graph with $n \approx 10^6$ nodes in $0.26$ seconds.
The plot also gives evidence towards \ipwl's linear work complexity.
Comparing with the {Edge-Switching} implementations, we find that \ipwl runs faster.
We can conclude that in this setting, the provably uniform \ipwl runs just as fast, if not faster, than the approximate solution.

\begin{figure}[!t]
 \centering
 \begin{minipage}{0.49\textwidth}
  \centering
  \includegraphics[scale=0.45]{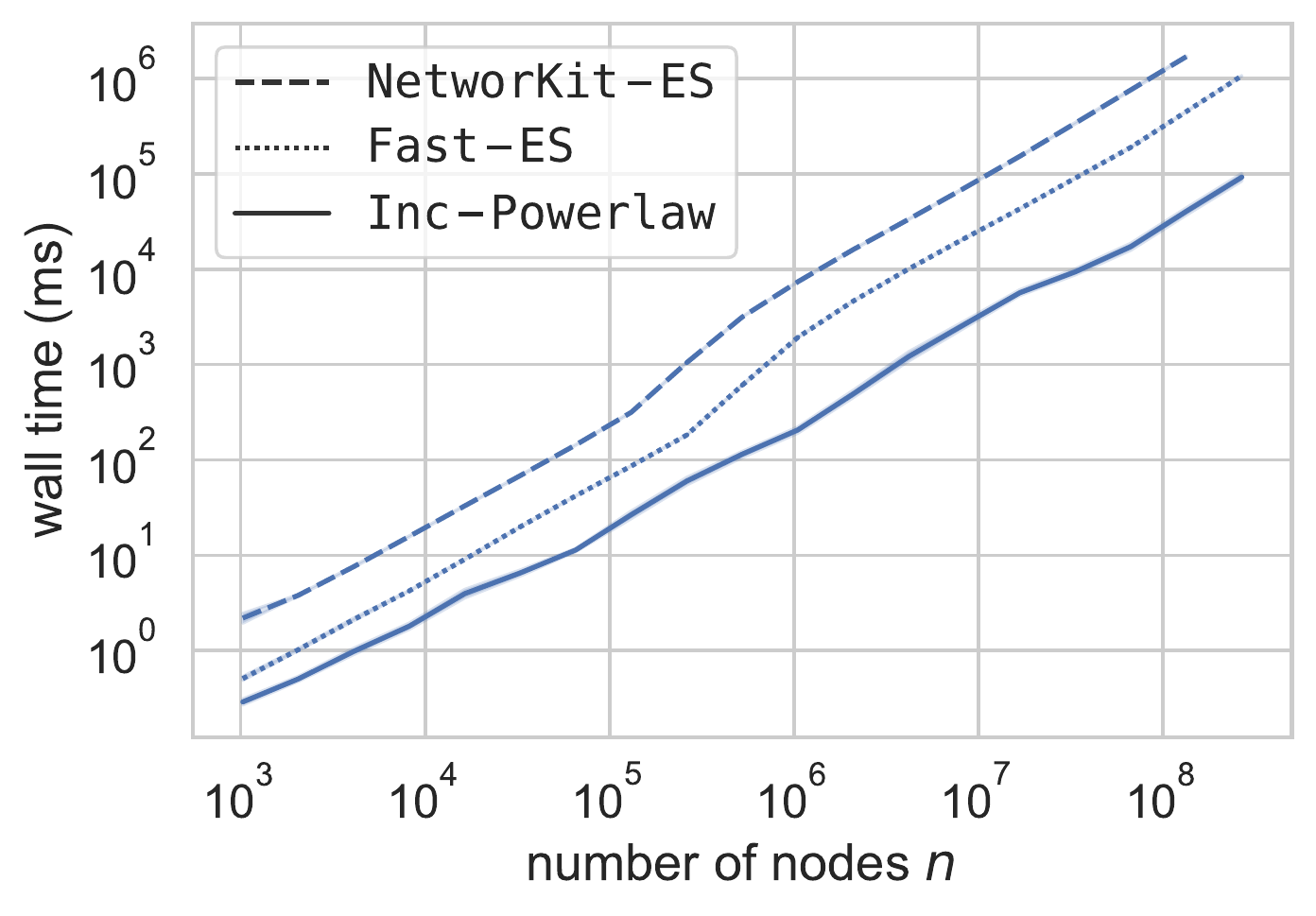}
  \caption{Average running-time of sequential \ipwl and Edge-Switching.}
  \label{fig:dep_n}
 \end{minipage}%
 \hfill
 \centering
 \begin{minipage}{0.49\textwidth}
  \centering
  \includegraphics[scale=0.45]{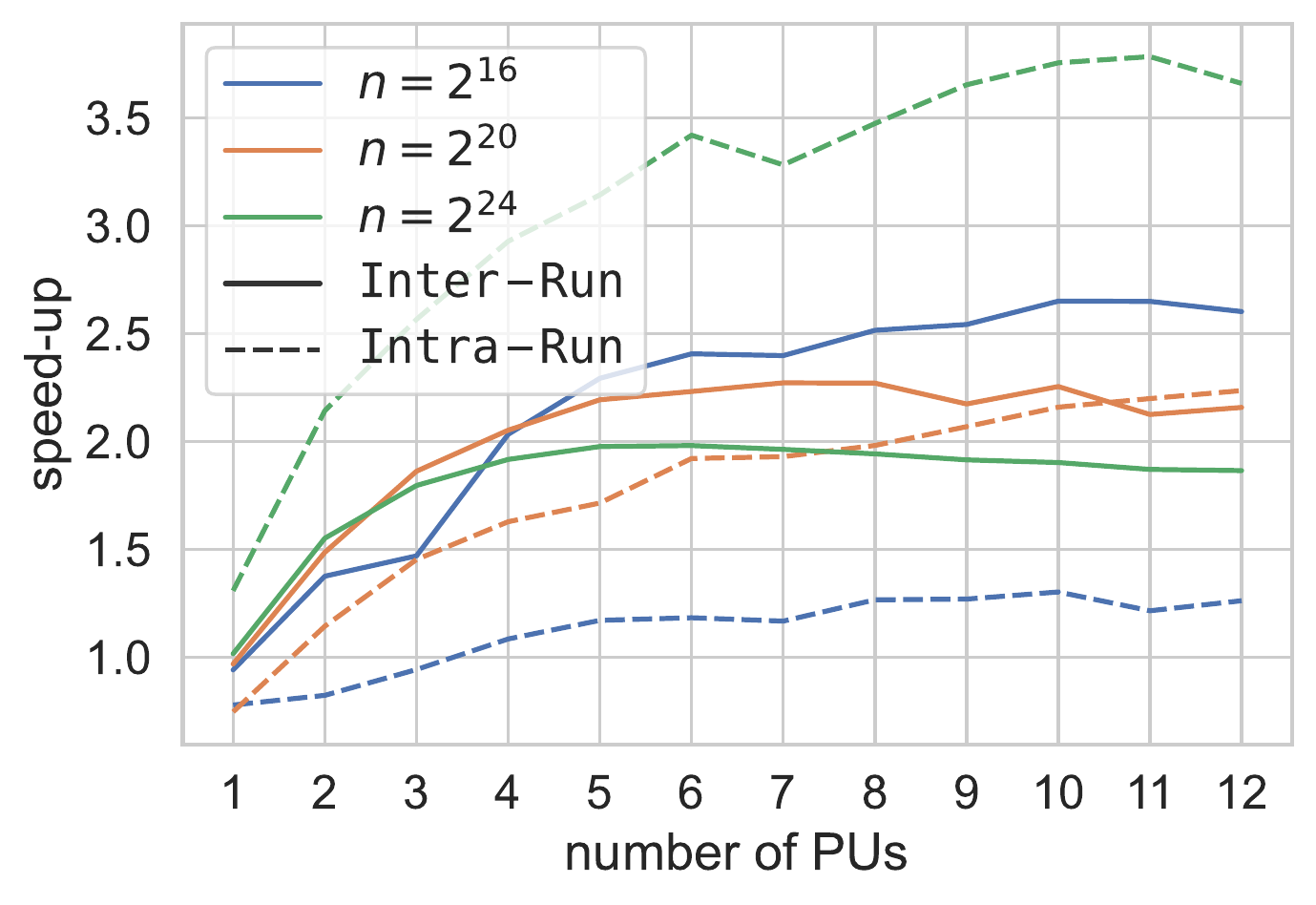}
  \caption{Speed-up of the parallel variants \parinter and \parintra over sequential \ipwl.}
  \label{fig:parallel_dep_n}
 \end{minipage}
\end{figure}

\paragraph*{Speed-up of the parallel variants}
\autoref{fig:parallel_dep_n} shows the speed-up of our \parinter and \parintra parallelizations over sequential \ipwl.
We generate degree sequences with $n = 2^k$ for $k \in \{16, 20, 24\}$, measure the average running-time of the parallel variants when using $1 \leq p \leq 12$ PUs (processor cores), and report the speed-up in the average running-time of the parallel variants over sequential \ipwl.

For $p=5$ and $n = 2^{24}$, we observe an \parinter parallelization speed-up of $2.0$;
more PUs yield diminishing returns as the speed-up is limited by the number of runs until a graph is accepted which is $3.0$ on average for the aforementioned parameters.
Another limiting factor is the fact that rejected runs stop prematurely.
Hence, the accepting run (i) requires on average more work and (ii) forms the critical path that cannot be accelerated by \parinter.

For the same $n$, \parintra achieves a speed-up of $3.8$ for $p = 11$ PUs; here, the remaining unparallelized sections limit the scalability as governed by Amdahl's law~\cite{DBLP:conf/isca/Rodgers85}.
Overall, \parinter yields a better speed-up if the the number of restarts is high (smaller $n$), whereas \parintra yields a better speed-up for larger $n$ if the overall running-time is dominated by generating the initial graph (see \autoref{table:restarts}).

\begin{table}[b]
  \centering
\caption{The average number of runs until acceptance and peak speed-ups as observed in \autoref{fig:parallel_dep_n}.}
  \begin{tabular}{ | c | c | c | c | c | }
    \hline
    $n$ & runs & \parinter & \parintra \\ \hline \hline
    $2^{16}$ & $3.9$ & $2.7$ for $p=10$ & $1.3$ for $p=10$ \\ \hline
    $2^{20}$ & $3.3$ & $2.3$ for $p=7$ & $2.2$ for $p=12$ \\ \hline
    $2^{24}$ & $3.0$ & $2.0$ for $p=5$ & $3.8$ for $p=11$ \\
    \hline
  \end{tabular}
  \label{table:restarts}
\end{table}

\paragraph*{Different values of the power-law exponent $\gamma$}
Next, we investigate the influence of the power-law exponent $\gamma$.
The guarantees on \ipwl's running-time only hold for sequences with $\gamma \gtrapprox 2.88102$, so we expect a superlinear running-time for $\gamma \leq 2.88$.
For $\gamma \geq 3$, the expected number of non-simple edges in the initial graph is much lower, so we expect the running-time to remain linear but with decreased constants.
\autoref{fig:dep_gamma} shows the average running-time of \ipwl for sequences for various $\gamma$. 

\begin{figure}[!t]
 \centering
 \begin{minipage}{0.49\textwidth}
  \centering
  \includegraphics[scale=0.45]{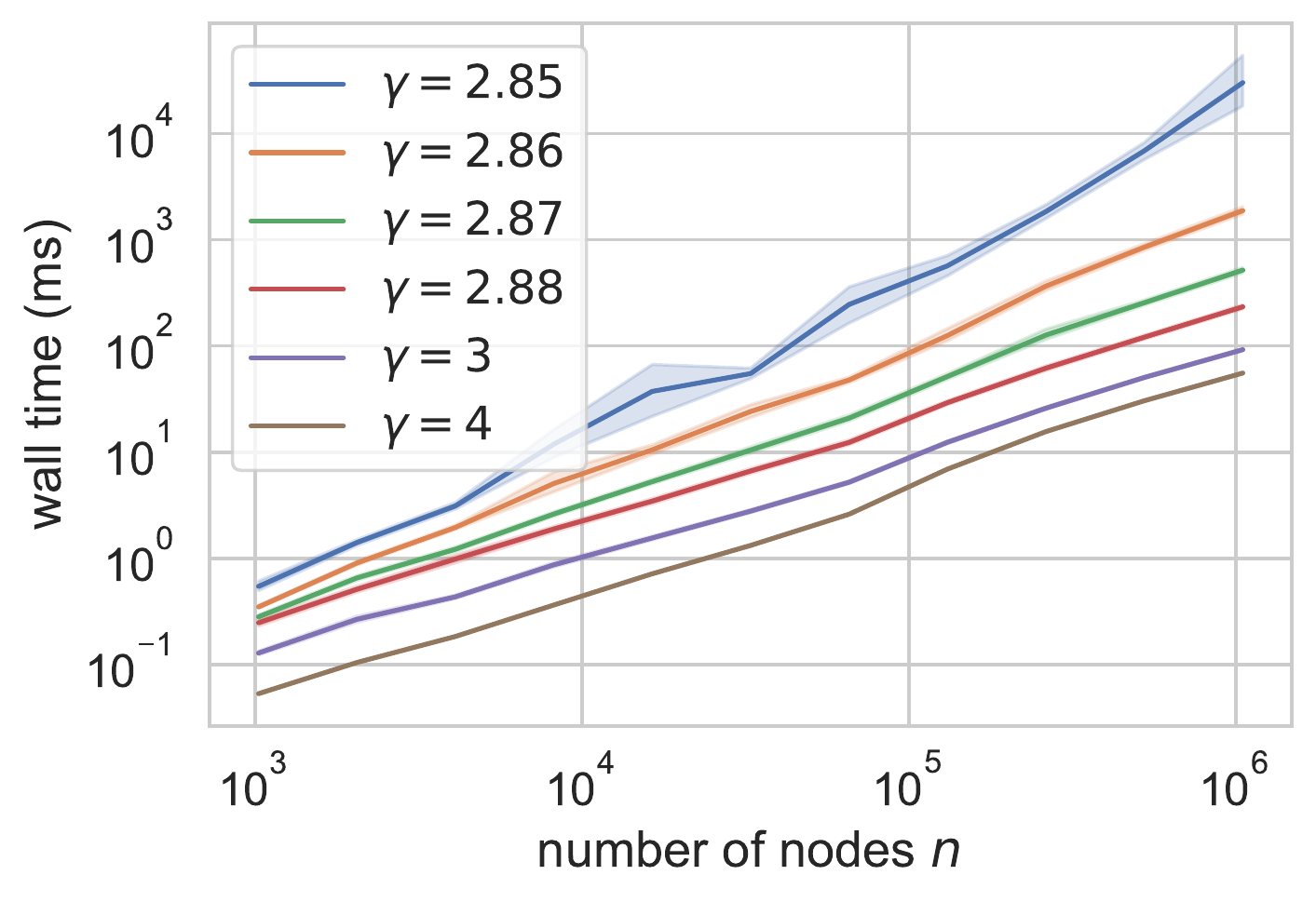}
  \caption{The average running-time of \ipwl in dependence of $n$ for different values of $\gamma$.}
  \label{fig:dep_gamma}
 \end{minipage}%
 \hfill
 \centering
 \begin{minipage}{0.49\textwidth}
  \centering
  \includegraphics[scale=0.45]{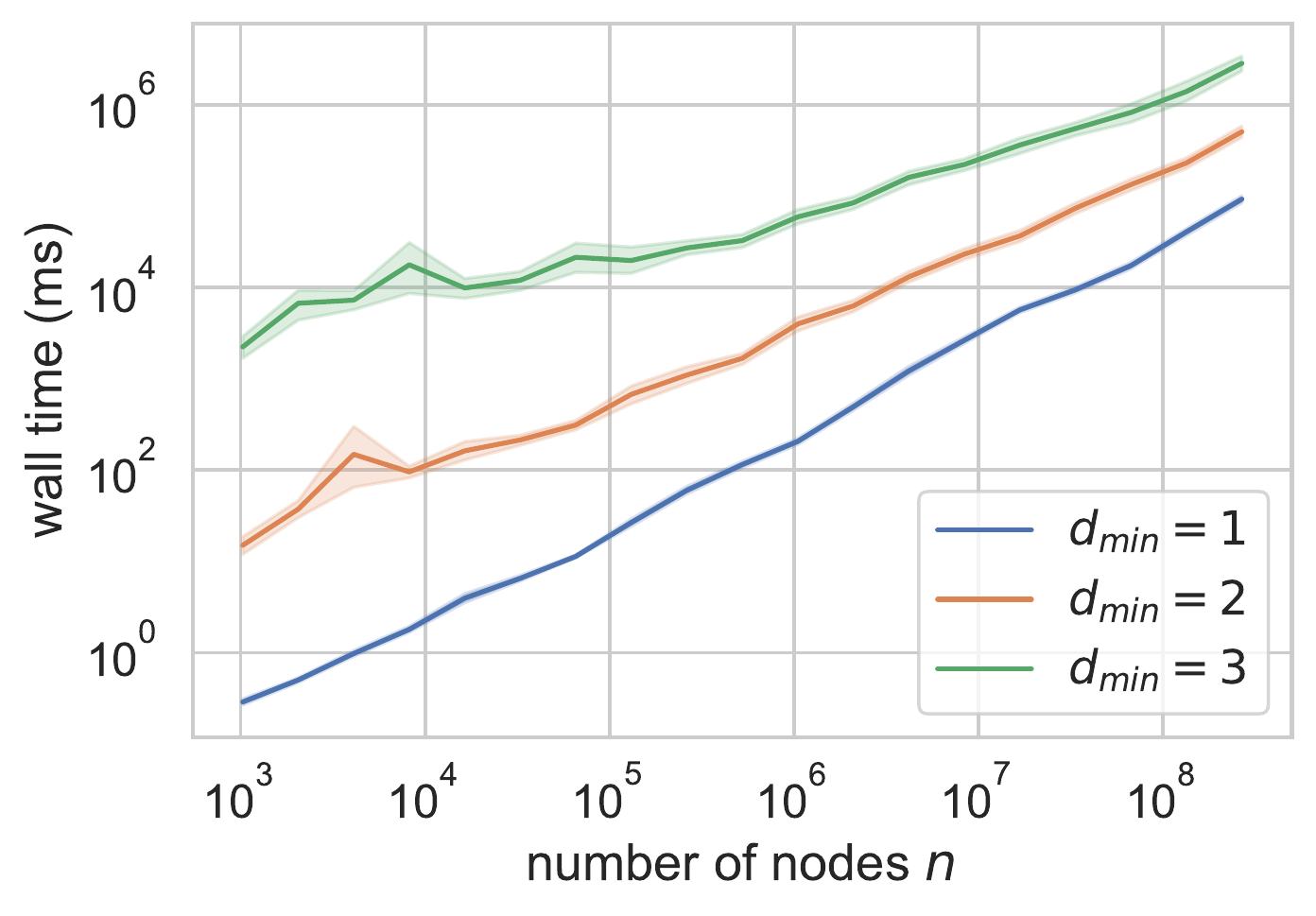}
  \caption{The average running-time of \ipwl for sequences with minimum degree $d_{min} \in \{1, 2, 3\}$.}
  \label{fig:dep_d-min}
 \end{minipage}
\end{figure}

\begin{table}[b]
  \caption{Average number of runs until acceptance and average number of switching steps in an accepting run.}
  \begin{tabular}{ | c | c | c | c | c | c | c | }
    \hline
    & \multicolumn{2}{c|}{$d_{min} = 1$} & \multicolumn{2}{c|}{$d_{min} = 2$} & \multicolumn{2}{c|}{$d_{min} = 3$} \\ \hline
    $n$ & runs & steps & runs & steps & runs & steps \\ \hline \hline
    $2^{16}$ & $3.9$ & $4.9$ & $53.0$ & $24.1$ & $1160.7$ & $51.0$ \\ \hline
    $2^{20}$ & $3.3$ & $10.2$ & $18.1$ & $49.6$ & $164.0$ & $110.8$ \\ \hline
    $2^{24}$ & $3.0$ & $19.0$ & $9.6$ & $92.8$ & $43.3$ & $208.5$ \\ \hline
    $2^{28}$ & $2.5$ & $32.3$ & $5.6$ & $161.6$ & $16.7$ & $375.3$ \\
    \hline
  \end{tabular}
  \centering
  \label{table:runs_steps_d_min}
\end{table}

For $\gamma \leq 2.88$, we observe an increase in the running-time.
The slope of the curve for $\gamma = 2.85$ also suggests that the running-time becomes non-linear for lower values of $\gamma$.
Overall, the requirement of $\gamma \gtrapprox 2.88102$ appears to be relatively strict.
In particular, we observe that the higher maximum degrees of sequences with $\gamma \leq 2.85$ greatly increase the rejection probability in Phases 1 and 2.

For $\gamma \geq 3$, the average running-time decreases somewhat but remains linear.
For these values of $\gamma$, we observe that the initial number of non-simple edges in the graph is small, and that the algorithm almost always accepts a graph on its first run, so the overall running-time approaches the time required to sample the initial graph with the configuration model.

\paragraph*{Higher average degrees}
The previously considered sequences drawn from an unscaled power-law distribution tend to have a rather small average degree of approximately $1.44$.
On the other hand, many observed networks feature higher average degrees~\cite{barabasi2016network,DBLP:conf/aaai/RossiA15}.
To study \ipwl on such networks, we sample degree sequences with minimum degree $d_{min} \in \{1, 2, 3\}$.
For $d_{min} = 2$ and $d_{min} = 3$, the average degree $\overline{d}$ of the sequences increases to $\overline{d}=3.39$ and $\overline{d}=5.44$ respectively.
We then let the implementation generate graphs for each choice of $d_{min} \in \{1, 2, 3\}$, and report the average time as a function of $n$ in \autoref{fig:dep_d-min}.

As a higher average degree increases the expected number of non-simple edges in the initial graph, we observe a significant increase in running-time.
For instance, for $n = 2^{20}$ we find that the average number of double-edges in the initial graph are $6.5$, $41.6$ and $98.8$ for $d_{min} = 1, 2$ and $3$, respectively, and the overall number of switching steps until a simple graph is obtained increases from $10.2$ for $d_{min}=1$ to $49.6$ for $d_{min}=2$ and to $110.8$ for $d_{min}=3$.
This in turn greatly increases the chance for a rejection to occur and the number of runs until a graph is accepted (see \autoref{table:runs_steps_d_min}).

However, for large values of $n \geq 2^{24}$ the effect of the higher average degrees on the running-time becomes less pronounced.
This is because the probability of a rejection at any step in the algorithm decreases quite fast with $n$, thus even if the number of switching steps increases, the number of runs decreases.
We can conclude that \ipwl is efficient when generating graphs that are either very sparse ($\overline{d} \lessapprox 5$) or very large ($n \gtrapprox 2^{24}$), but the algorithm is much less efficient when generating small to medium sized graphs ($n \lessapprox 2^{24}$) with medium average degree ($\overline{d} \gtrapprox 5$).

\paragraph*{Speed-up of \parinter for higher average degrees}
While \ipwl's sequential work increases with a higher average degree, so do the number of independent runs that can be parallelized by \parinter. \autoref{fig:parallel_d-min} shows the speedup of \parinter over sequential \ipwl for sequences with $d_{min} = 2$ when using $2 \leq p \leq 24$ PUs and $d_{min} = 3$ using $2 \leq p \leq 32$ PUs.
For $n = 2^{20}$ nodes, \parinter yields a speed-up of $6.4$ with $p=14$ PUs for $d_{min} = 2$ and $12.8$ for $p=31$ PUs for $d_{min} = 3$ (see \autoref{table:restarts_d_min}).

\begin{figure}[t]
 \centering
 \begin{subfigure}{0.49\textwidth}
 \resizebox{\textwidth}{!}{
  	\includegraphics{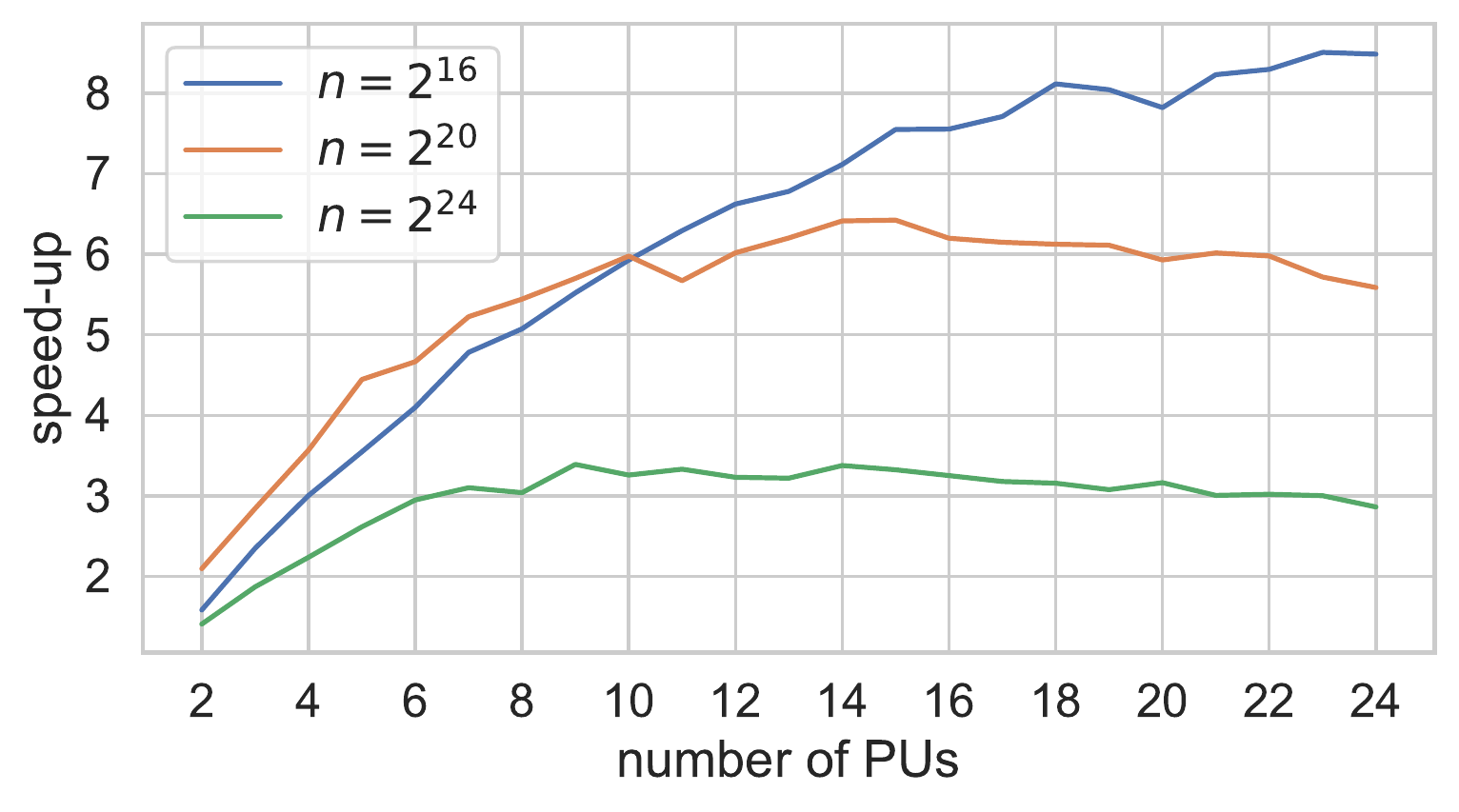}
  }
  \caption{\centering Speed-ups for $d_{min} = 2$.}
 \end{subfigure}
 \hfill
 \centering
 \begin{subfigure}{0.49\textwidth}
 \resizebox{\textwidth}{!}{
  \includegraphics{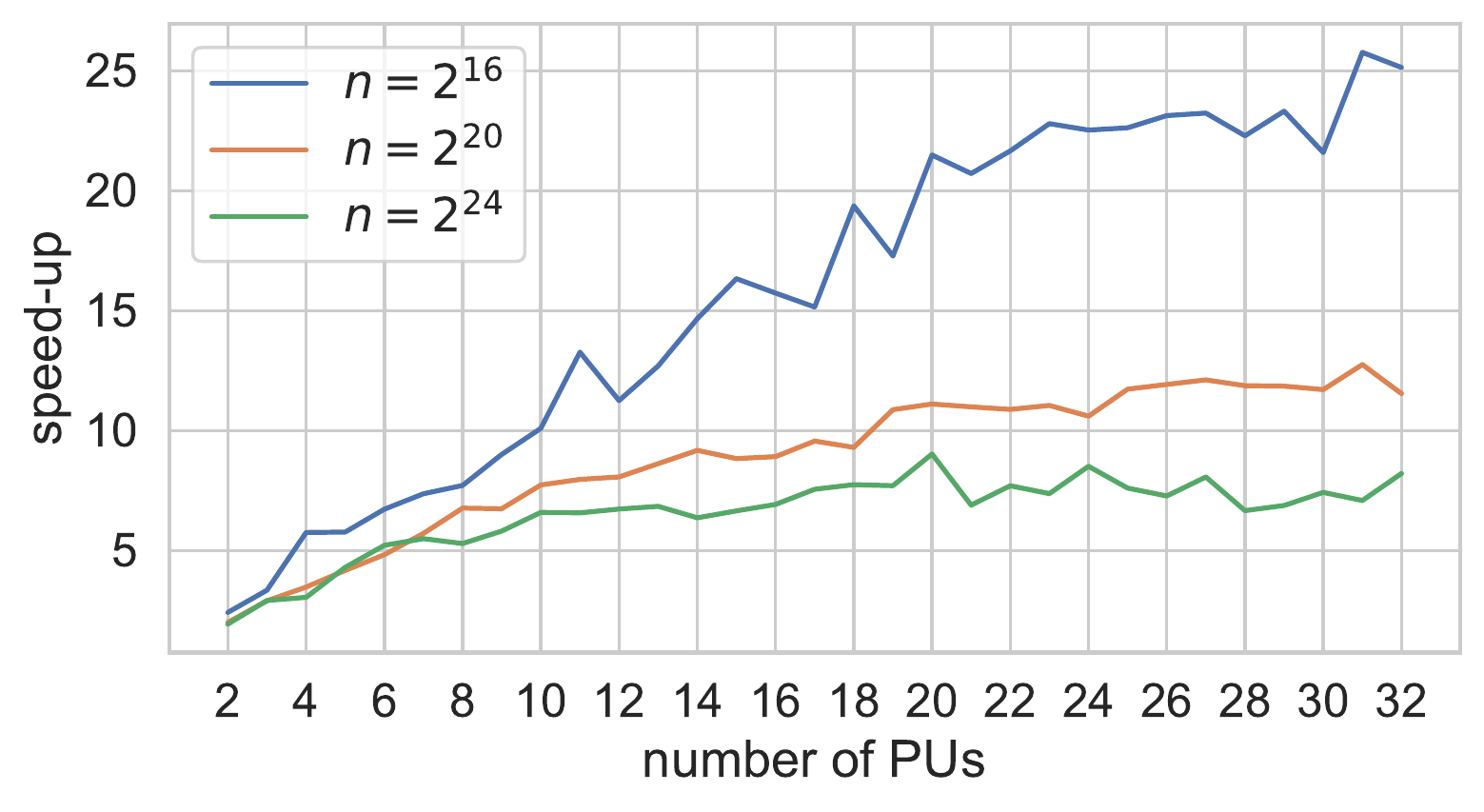}
  }
  \caption{\centering Speed-ups for $d_{min} = 3$.}
 \end{subfigure}
 \vspace{1em}
 \caption{\centering Speed-up of \parinter for $d_{min} \in \{2, 3\}$.}
 \label{fig:parallel_d-min}
\end{figure}

\begin{table}[t]
  \vspace{2em}
  \caption{Average number of runs until acceptance and peak speed-ups as observed in \autoref{fig:parallel_d-min}.}
  \begin{tabular}{ | c | c | c | c | c | }
    \hline
    & \multicolumn{2}{c|}{$d_{min} = 2$} & \multicolumn{2}{c|}{$d_{min} = 3$} \\ \hline
    $n$ & runs & speedup & runs & speedup \\ \hline \hline
    $2^{16}$ & $53.0$ & $8.5$ for $p=24$ & $1160.7$ & $25.8$ for $p=31$ \\ \hline
    $2^{20}$ & $18.1$ & $6.4$ for $p=14$ & $164.0$ & $12.8$ for $p=31$ \\ \hline
    $2^{24}$ & $9.6$ & $3.4$ for $p=8$ & $43.3$ & $9.0$ for $p=20$ \\
    \hline
  \end{tabular}
  \centering
  \label{table:restarts_d_min}
 \vspace{-1em}
\end{table}

As expected, we can achieve a higher speed-up for higher $d_{min}$, so we can partially mitigate the increase in running-time by taking advantage of the higher parallelizability.
On the other hand, we still experience the limited scaling due to accepting runs being slower than rejecting runs.

\section{Conclusions}
For the first time, we provide a complete description of \ipwl which builds on and extends previously known  results~\cite{DBLP:conf/focs/ArmanGW19,DBLP:conf/soda/GaoW18}.
To the best of our knowledge, \ipwl is the first practical implementation to sample provably uniform graphs from prescribed power-law-bounded degree sequences with $\gamma \ge 2.88102$.
In an empirical study, we find that \ipwl is very efficient for small average degrees;
for larger average degrees, we observe significantly increased constants in \ipwl's running-time which are partially mitigated by an improved parallelizability.

While the expected running-time of \ipwl is asymptotically optimal, we expect practical improvements for higher average degrees by improving the acceptance probability in Phases 3, 4 and 5 of the algorithm (\eg by finding tighter lower bounds or by adding new switchings).
It is also possible that the requirement on $\gamma$ could be lowered;
our experiments indicate that the acceptance probability in Phases 1 and 2 should be improved to this end.
Our measurements also suggest that more fine-grained parallelism may be necessary to accelerate accepting runs.

\bibliography{pld-bibliography}

\begin{thebibliography}{10}

\bibitem{DBLP:conf/sc/AlamKVM16}
M.~Alam, M.~Khan, A.~Vullikanti, and M.~Marathe.
\newblock An efficient and scalable algorithmic method for generating large:
  scale random graphs.
\newblock In {\em {SC}}. {IEEE} Computer Society, 2016.
\newblock \href {https://doi.org/10.1109/SC.2016.31}
  {\path{doi:10.1109/SC.2016.31}}.

\bibitem{DBLP:conf/focs/ArmanGW19}
A.~Arman, P.~Gao, and N.~C. Wormald.
\newblock Fast uniform generation of random graphs with given degree sequences.
\newblock In {\em {FOCS}}. {IEEE} Computer Society, 2019.
\newblock \href {https://doi.org/10.1109/FOCS.2019.00084}
  {\path{doi:10.1109/FOCS.2019.00084}}.

\bibitem{DBLP:journals/corr/abs-2009-13569}
M.~Axtmann, S.~Witt, D.~Ferizovic, and P.~Sanders.
\newblock Engineering in-place (shared-memory) sorting algorithms.
\newblock {\em CoRR}, abs/2009.13569, 2020.

\bibitem{barabasi2016network}
A.~Barab{\'a}si et~al.
\newblock {\em Network science}.
\newblock Cambridge university press, 2016.

\bibitem{DBLP:journals/algorithmica/BayatiKS10}
M.~Bayati, J.~H. Kim, and A.~Saberi.
\newblock A sequential algorithm for generating random graphs.
\newblock {\em Algorithmica}, 58(4), 2010.
\newblock \href {https://doi.org/10.1007/s00453-009-9340-1}
  {\path{doi:10.1007/s00453-009-9340-1}}.

\bibitem{bekessy1972asymptotic}
A.~B{\'e}k{\'e}ssy, P.~B{\'e}k{\'e}ssy, and J.~Koml{\'o}s.
\newblock Asymptotic enumeration of regular matrices.
\newblock {\em Stud. Sci. Math. Hungar.}, 7, 1972.

\bibitem{DBLP:journals/jct/BenderC78}
E.~A. Bender and E.~R. Canfield.
\newblock The asymptotic number of labeled graphs with given degree sequences.
\newblock {\em J. Comb. Theory, Ser. {A}}, 24(3), 1978.
\newblock \href {https://doi.org/10.1016/0097-3165(78)90059-6}
  {\path{doi:10.1016/0097-3165(78)90059-6}}.

\bibitem{DBLP:conf/icpp/BhuiyanCKM14}
M.~H. Bhuiyan, J.~Chen, M.~Khan, and M.~V. Marathe.
\newblock Fast parallel algorithms for edge-switching to achieve a target visit
  rate in heterogeneous graphs.
\newblock In {\em 43rd Int. Conf. on Parallel Processing, {ICPP} 2014}.
  {{IEEE}} Computer Society, 2014.
\newblock \href {https://doi.org/10.1109/ICPP.2014.15}
  {\path{doi:10.1109/ICPP.2014.15}}.

\bibitem{DBLP:journals/ejc/Bollobas80}
B.~Bollob{\'{a}}s.
\newblock A probabilistic proof of an asymptotic formula for the number of
  labelled regular graphs.
\newblock {\em Eur. J. Comb.}, 1(4), 1980.
\newblock \href {https://doi.org/10.1016/S0195-6698(80)80030-8}
  {\path{doi:10.1016/S0195-6698(80)80030-8}}.

\bibitem{bollobas1985random}
B.~Bollob{\'{a}}s.
\newblock {\em Random graphs}.
\newblock Academic Press, 1985.

\bibitem{DBLP:conf/esa/CarstensH0PTW18}
C.~J. Carstens, M.~Hamann, U.~Meyer, M.~Penschuck, H.~Tran, and D.~Wagner.
\newblock Parallel and {I/O}-efficient randomisation of massive networks using
  global curveball trades.
\newblock In {\em {ESA}}, volume 112 of {\em LIPIcs}, 2018.
\newblock \href {https://doi.org/10.4230/LIPIcs.ESA.2018.11}
  {\path{doi:10.4230/LIPIcs.ESA.2018.11}}.

\bibitem{chung2002connected}
F.~Chung and L.~Lu.
\newblock Connected components in random graphs with given expected degree
  sequences.
\newblock {\em Annals of Combinatorics}, 6(2), 2002.
\newblock \href {https://doi.org/10.1007/pl00012580}
  {\path{doi:10.1007/pl00012580}}.

\bibitem{DBLP:journals/cpc/CooperDG07}
C.~Cooper, M.~E. Dyer, and C.~S. Greenhill.
\newblock Sampling regular graphs and a peer-to-peer network.
\newblock {\em Comb. Probab. Comput.}, 16(4), 2007.
\newblock \href {https://doi.org/10.1017/S0963548306007978}
  {\path{doi:10.1017/S0963548306007978}}.

\bibitem{DBLP:journals/jpdc/FunkeLMPSSSL19}
D.~Funke, S.~Lamm, U.~Meyer, M.~Penschuck, P.~Sanders, C.~Schulz, D.~Strash,
  and M.~von Looz.
\newblock Communication-free massively distributed graph generation.
\newblock {\em J. Parallel Distributed Comput.}, 131, 2019.
\newblock \href {https://doi.org/10.1016/j.jpdc.2019.03.011}
  {\path{doi:10.1016/j.jpdc.2019.03.011}}.

\bibitem{DBLP:conf/focs/GaoW15}
P.~Gao and N.~C. Wormald.
\newblock Uniform generation of random regular graphs.
\newblock In {\em {FOCS}}. {IEEE} Computer Society, 2015.
\newblock \href {https://doi.org/10.1109/FOCS.2015.78}
  {\path{doi:10.1109/FOCS.2015.78}}.

\bibitem{DBLP:conf/soda/GaoW18}
P.~Gao and N.~C. Wormald.
\newblock Uniform generation of random graphs with power-law degree sequences.
\newblock In {\em {SODA}}. {SIAM}, 2018.
\newblock \href {https://doi.org/10.1137/1.9781611975031.114}
  {\path{doi:10.1137/1.9781611975031.114}}.

\bibitem{DBLP:conf/alenex/GkantsidisMMZ03}
C.~Gkantsidis, M.~Mihail, and E.~W. Zegura.
\newblock The {Markov} chain simulation method for generating connected power
  law random graphs.
\newblock In {\em {ALENEX}}. {SIAM}, 2003.

\bibitem{gotelli1996null}
N.~J. Gotelli and G.~R. Graves.
\newblock {\em Null models in ecology}.
\newblock Smithsonian Institution, 1996.

\bibitem{DBLP:conf/soda/Greenhill15}
C.~S. Greenhill.
\newblock The switch {Markov} chain for sampling irregular graphs (extended
  abstract).
\newblock In {\em {SODA}}. {SIAM}, 2015.
\newblock \href {https://doi.org/10.1137/1.9781611973730.103}
  {\path{doi:10.1137/1.9781611973730.103}}.

\bibitem{DBLP:journals/jea/HamannMPTW18}
M.~Hamann, U.~Meyer, M.~Penschuck, H.~Tran, and D.~Wagner.
\newblock I/o-efficient generation of massive graphs following the \emph{LFR}
  benchmark.
\newblock {\em {ACM} J. Exp. Algorithmics}, 23, 2018.
\newblock \href {https://doi.org/10.1145/3230743} {\path{doi:10.1145/3230743}}.

\bibitem{janson2009probability}
S.~Janson.
\newblock The probability that a random multigraph is simple.
\newblock {\em Combinatorics, Probability and Computing}, 18(1-2), 2009.

\bibitem{DBLP:journals/tcs/JerrumS90}
M.~Jerrum and A.~Sinclair.
\newblock Fast uniform generation of regular graphs.
\newblock {\em Theor. Comput. Sci.}, 73(1), 1990.
\newblock \href {https://doi.org/10.1016/0304-3975(90)90164-D}
  {\path{doi:10.1016/0304-3975(90)90164-D}}.

\bibitem{DBLP:journals/rsa/KannanTV99}
R.~Kannan, P.~Tetali, and S.~S. Vempala.
\newblock Simple {Markov}-chain algorithms for generating bipartite graphs and
  tournaments.
\newblock {\em Random Struct. Algorithms}, 14(4), 1999.

\bibitem{DBLP:journals/combinatorica/KimV06}
J.~H. Kim and V.~H. Vu.
\newblock Generating random regular graphs.
\newblock {\em Comb.}, 26(6), 2006.
\newblock \href {https://doi.org/10.1007/s00493-006-0037-7}
  {\path{doi:10.1007/s00493-006-0037-7}}.

\bibitem{Lancichinetti2009}
A.~Lancichinetti and S.~Fortunato.
\newblock Benchmarks for testing community detection algorithms on directed and
  weighted graphs with overlapping communities.
\newblock {\em Phys. Rev. E}, 80(1), 2009.
\newblock \href {https://doi.org/10.1103/physreve.80.016118}
  {\path{doi:10.1103/physreve.80.016118}}.

\bibitem{DBLP:journals/tomacs/Lemire19}
D.~Lemire.
\newblock Fast random integer generation in an interval.
\newblock {\em {ACM} Trans. Model. Comput. Simul.}, 29(1), 2019.
\newblock \href {https://doi.org/10.1145/3230636} {\path{doi:10.1145/3230636}}.

\bibitem{DBLP:conf/sigcomm/MahadevanKFV06}
P.~Mahadevan, D.~V. Krioukov, K.~R. Fall, and V.~Vahdat.
\newblock Systematic topology analysis and generation using degree
  correlations.
\newblock In {\em {SIGCOMM}}. {ACM}, 2006.

\bibitem{DBLP:journals/jal/McKayW90}
B.~D. McKay and N.~C. Wormald.
\newblock Uniform generation of random regular graphs of moderate degree.
\newblock {\em J. Algorithms}, 11(1), 1990.
\newblock \href {https://doi.org/10.1016/0196-6774(90)90029-E}
  {\path{doi:10.1016/0196-6774(90)90029-E}}.

\bibitem{DBLP:conf/waw/MillerH11}
J.~C. Miller and A.~A. Hagberg.
\newblock Efficient generation of networks with given expected degrees.
\newblock In {\em {WAW}}, volume 6732 of {\em Lecture Notes in Computer
  Science}, 2011.
\newblock \href {https://doi.org/10.1007/978-3-642-21286-4_10}
  {\path{doi:10.1007/978-3-642-21286-4_10}}.

\bibitem{Milo824}
R.~Milo.
\newblock Network motifs: Simple building blocks of complex networks.
\newblock {\em Science}, 298(5594), 2002.
\newblock \href {https://doi.org/10.1126/science.298.5594.824}
  {\path{doi:10.1126/science.298.5594.824}}.

\bibitem{Milo2003}
R.~{Milo}, N.~{Kashtan}, S.~{Itzkovitz}, M.~E.~J. {Newman}, and U.~{Alon}.
\newblock {On the uniform generation of random graphs with prescribed degree
  sequences}.
\newblock December 2003.
\newblock \href {http://arxiv.org/abs/cond-mat/0312028}
  {\path{arXiv:cond-mat/0312028}}.

\bibitem{MorenoPN18}
S.~Moreno, J.~J. Pfeiffer~III, and J.~Neville.
\newblock Scalable and exact sampling method for probabilistic generative graph
  models.
\newblock {\em Data Mining and Knowledge Discovery}, 32(6), 2018.
\newblock \href {https://doi.org/10.1007/s10618-018-0566-x}
  {\path{doi:10.1007/s10618-018-0566-x}}.

\bibitem{DBLP:books/ox/Newman10}
M.~E.~J. Newman.
\newblock {\em Networks: An Introduction}.
\newblock Oxford University Press, 2010.
\newblock \href {https://doi.org/10.1093/ACPROF:OSO/9780199206650.001.0001}
  {\path{doi:10.1093/ACPROF:OSO/9780199206650.001.0001}}.

\bibitem{DBLP:journals/corr/abs-2003-00736}
M.~Penschuck, U.~Brandes, M.~Hamann, S.~Lamm, U.~Meyer, I.~Safro, P.~Sanders,
  and C.~Schulz.
\newblock Recent advances in scalable network generation.
\newblock {\em CoRR}, abs/2003.00736, 2020.

\bibitem{DBLP:conf/waw/RayPS12}
J.~Ray, A.~Pinar, and C.~Seshadhri.
\newblock Are we there yet? when to stop a markov chain while generating random
  graphs.
\newblock In {\em {WAW}}, volume 7323 of {\em Lecture Notes in Computer
  Science}, 2012.
\newblock \href {https://doi.org/10.1007/978-3-642-30541-2_12}
  {\path{doi:10.1007/978-3-642-30541-2_12}}.

\bibitem{DBLP:conf/isca/Rodgers85}
D.~P. Rodgers.
\newblock Improvements in multiprocessor system design.
\newblock In {\em {ISCA}}. {IEEE} Computer Society, 1985.

\bibitem{DBLP:conf/aaai/RossiA15}
R.~A. Rossi and N.~K. Ahmed.
\newblock The network data repository with interactive graph analytics and
  visualization.
\newblock In {\em {AAAI}}. {AAAI} Press, 2015.

\bibitem{DBLP:journals/ipl/Sanders98}
P.~Sanders.
\newblock Random permutations on distributed, external and hierarchical memory.
\newblock {\em Inf. Process. Lett.}, 67(6), 1998.

\bibitem{DBLP:journals/snam/SchlauchHZ15}
W.~E. Schlauch, E.~\'{A}. Horv{\'{a}}t, and K.~A. Zweig.
\newblock Different flavors of randomness: comparing random graph models with
  fixed degree sequences.
\newblock {\em Social Network Analysis and Mining}, 5(1), 2015.
\newblock \href {https://doi.org/10.1007/s13278-015-0267-z}
  {\path{doi:10.1007/s13278-015-0267-z}}.

\bibitem{singler2008gnu}
J.~Singler and B.~Konsik.
\newblock The {GNU} {libstdc++} parallel mode: software engineering
  considerations.
\newblock In {\em Int. workshop on Multicore software eng.}, 2008.

\bibitem{DBLP:conf/alenex/StantonP11}
I.~Stanton and A.~Pinar.
\newblock Sampling graphs with a prescribed joint degree distribution using
  {Markov} chains.
\newblock In {\em {ALENEX}}. {SIAM}, 2011.

\bibitem{DBLP:journals/netsci/StaudtSM16}
C.~L. Staudt, A.~Sazonovs, and H.~Meyerhenke.
\newblock Networkit: {A} tool suite for large-scale complex network analysis.
\newblock {\em Netw. Sci.}, 4(4), 2016.
\newblock \href {https://doi.org/10.1017/nws.2016.20}
  {\path{doi:10.1017/nws.2016.20}}.

\bibitem{DBLP:journals/cpc/StegerW99}
A.~Steger and N.~C. Wormald.
\newblock Generating random regular graphs quickly.
\newblock {\em Comb. Probab. Comput.}, 8(4), 1999.

\bibitem{strona2014fast}
G.~Strona, D.~Nappo, F.~Boccacci, S.~Fattorini, and J.~San-Miguel-Ayanz.
\newblock A fast and unbiased procedure to randomize ecological binary matrices
  with fixed row and column totals.
\newblock {\em Nature communications}, 5(1), 2014.

\bibitem{tinhofer1979generation}
G.~Tinhofer.
\newblock On the generation of random graphs with given properties and known
  distribution.
\newblock {\em Appl. Comput. Sci., Ber. Prakt. Inf}, 13:265--297, 1979.

\bibitem{verhelst2008efficient}
N.~D. Verhelst.
\newblock An efficient {MCMC} algorithm to sample binary matrices with fixed
  marginals.
\newblock {\em Psychometrika}, 73(4), 2008.

\bibitem{DBLP:journals/compnet/VigerL16}
F.~Viger and M.~Latapy.
\newblock Efficient and simple generation of random simple connected graphs
  with prescribed degree sequence.
\newblock {\em J. Complex Networks}, 4(1), 2016.
\newblock \href {https://doi.org/10.1093/comnet/cnv013}
  {\path{doi:10.1093/comnet/cnv013}}.

\bibitem{DBLP:journals/corr/Zhao13b}
J.~Y. Zhao.
\newblock Expand and contract: Sampling graphs with given degrees and other
  combinatorial families.
\newblock 2013.
\newblock \href {http://arxiv.org/abs/1308.6627} {\path{arXiv:1308.6627}}.

\end{thebibliography}

\appendix

\section{Summary of symbols used}
\label{sec:apx-table-defs}
\begin{table}[!h]
	\label{tab:symbols}
 	\definecolor{tablegray}{gray}{0.9}
	\rowcolors{1}{white}{tablegray}
	\centering
	\begin{tabular}{l|l|p{0.8\textwidth}}
    \toprule
    Symbol & Section & Remark \\
    \midrule
    $[x]_k$ & \ref{subsec:notation} & $k$-th factorial moment $[x]_k = \prod_{i=0}^{k-1} (x-i)$ \\
$ij$, $\{i, j\}$ & \ref{subsec:notation} & edge connecting nodes $i$ and $j$ \\
$(i, j)$ & \ref{subsec:notation} & pair connecting nodes $i$ and $j$ \\
$m_{i, j}$ & \ref{subsec:notation} & multiplicity of edge $ij$, number of pairs $(i, j)$ \\
$v_1v_2v_3$ & \ref{subsec:phase3} & two-star centered at $v_1$ \\
$v_1v_2v_3v_4$ & \ref{subsec:phase4} & three-star centered at $v_1$ \\
$\degseq$, $d_i$ & \ref{subsec:notation} & degree sequence $\degseq = (d_1, \dots, d_n) \in \mathbb N ^ n$ \\
$\mathcal{G}(\degseq)$ & \ref{subsec:notation} & set of simple graphs matching degree sequence $\degseq$ \\
plib & \ref{subsec:notation}	& power-law distribution-bounded \\
$\gamma$ & \ref{subsec:notation}	& power-law exponent \\
$h$ & \ref{subsec:notation} & number of heavy nodes $h = n^{1 - \delta (\gamma - 1)}$ where $\frac{1}{2 \gamma - 3} < \delta < \frac{2 - 3/(\gamma - 1)}{4 - \gamma}$ \\
$M_k$ & \ref{subsec:notation} & $M_k = \sum_{i=1}^n [d_i]_k$, $k$-th moment of the degree distribution  \\
$H_k$ & \ref{subsec:notation} & $H_k = \sum_{i=1}^h [d_i]_k$ \\
$L_k$ & \ref{subsec:notation} & $L_k = M_k - H_k$\\
$W_i$ & \ref{subsec:p12preconditions} & sum of the multiplicities of heavy multi-edges incident to $i$ \\
$W_{i,j}$ & \ref{subsec:p12preconditions} & $W_{i,j} = W_i + 2 m_{i,i} - m_{i,j}$\\
$f(G)$ & \ref{subsec:phase1} & number of valid switchings on $G$\\
$\overline{f}(G)$ & \ref{subsec:phase1} & upper bound on $f(G)$\\
$\underline{f}(G)$ & \ref{subsec:phase1} & lower bound on $f(G)$\\
$b(G')$ & \ref{subsec:phase1} & number of valid switchings that produce $G'$\\
$\overline{b}(G')$ & \ref{subsec:phase1} & upper bound on $b(G')$\\
$\underline{b}(G')$ & \ref{subsec:phase1} & lower bound on $b(G')$\\
$m_l(G)$ & \ref{subsec:phase3prec} & number of light single loops in $G$ \\
$m_t(G)$ & \ref{subsec:phase3prec} & number of light triple-edges in $G$ \\
$m_d(G)$ & \ref{subsec:phase3prec} & number of light double-edges in $G$ \\
$b(G', \overline{V}_i)$ & \ref{subsec:phase3} & number of structures in $G'$ matching a valid switching that creates $\overline{V}_i$ \\
$\underline{b}(G';i)$ & \ref{subsec:phase3} & lower bound on $b(G', \overline{V}_i)$ \\
$A_2$ & \ref{subsec:phase3} & let $A_2 = \sum\nolimits_{i=1}^{d_1} d_i$, then for any node $v$, $A_2$ is an upper bound on the number of simple two-stars where $v$ is one of the outer nodes\\
$\tau$ & \ref{subsec:phase4} & switching type chosen in iterations of Phase 4 and 5 \\
$\rho_\tau$ & \ref{subsec:phase4} & probability of chosing type $\tau$ \\
$\varepsilon$ & \ref{subsec:phase4} & $\varepsilon = 28 M_2^2 / M_1^3$, upper bound on the probability of choosing a type $t_a, t_b$, or $t_c$ switching in Phase 4 \\
$k$ & \ref{subsec:phase4} & number of additional pairs created by Phase 4 or 5 booster switchings \\
$B_k$ & \ref{subsec:phase4} & let $B_k = \sum_{i=1}^{d_1} [d_{h+i}]_k$, then for any node $v$, $B_k$ is an upper bound on the number of simple light $k$-stars where $v$ is one of the outer nodes\\
$\xi$ & \ref{subsec:phase5} & $\xi = 32 M_2^2 / M_1^3 + 36 M_4 L_4 / M_2 L_2 M_1^2 + 32 M_3^2 / M_1^4$, upper bound on the probability of choosing a switching type other than type $d$ or $(1,0,0)$ in Phase 5 \\
		\bottomrule
	\end{tabular}
\end{table}

\section{Additional Proofs}
\label{sec:apx-proofs}
\subsection{Correctness proofs of lower bounds}
\label{subsec:proofbounds}
For Phases 3, 4 and 5, we use new lower bounds on the number of structures in the graph created by a valid switching.

For Phase 3, we factorize the lower bound on the number of inverse $l$-switchings used in \pld to obtain two new lower bounds $\underline{b}_l(G';0)$ and $\underline{b}_l(G';1)$.
\begin{lemma}\label{lemma:lb3}
Let $\mathcal{S}$ be the class of graphs with $m_t$ light triple-edges, $m_d$ light double-edges and $m_l$ light single loops (and no other non-simple edges).
For all $G \in \mathcal{S}$, and all light simple two-stars $v_1 v_2 v_3$ in $G$ that are created by a valid $l$-switching, we have
\begin{align}
 \underline{b}_l(\mathcal{S};0) &\leq b_l(G, \emptyset)
 \\
 \underline{b}_l(\mathcal{S};1) &\leq b_l(G, v_1 v_2 v_3).
\end{align}
\end{lemma}

\begin{proof}
We have $\underline{b}_l(\mathcal{S};0) = L_2 - 12 m_t d_h - 8 m_d d_h - m_l d_h^2$, and $b_l(G, \emptyset)$ is equal to the number of light simple ordered two-stars in $G$.
We now show that $\underline{b}_l(\mathcal{S};0)$ is a lower bound on $b_l(G, \emptyset)$.
First, each graph $G$ matching the sequence contains exactly $L_2$ light ordered two-stars.
We then overestimate the number of two-stars that are not simple, and subtract this from $L_2$: a two-star $v_1 v_2 v_3$ is not simple if one of the edges $v_1 v_2$ or $v_1 v_3$ is a triple-edge, a double-edge or a loop.
There are at most $12 m_t d_h$ that contain a triple-edge, as there are $m_t$ ordered triple-edges ($6m_t$ ordered pairs), at most $d_h$ choices for the remaining node of the two-star (any light node has degree smaller than $d_h$), and $2$ ways to combine the selected pairs into the two-star as shown in \autoref{fig:lswitching}.
Similarly, there are at most $8 m_d d_h$ two-stars that contain a double-edge, as there are $m_d$ double-edges ($4m_d$ ordered pairs), at most $d_h$ choices for the remaining node and $2$ ways to combine the selected pairs into the two-star, and there are at most $m_l d_h^2$ two-stars that contain a loop, as there are $m_l$ loops and at most $d_h^2$ choices for the outer nodes of the two-star.

For the second bound, we have $\underline{b}_l(\mathcal{S};1) = M_1 - 6 m_t - 4 m_d - 2 m_l - 2 A_2 - 4 d_1 - 2 d_h$ and $b_l(G, v_1 v_2 v_3)$ is set to the number of simple ordered pairs $(v_4, v_5)$ that (a) do share nodes with the two-star $v_1 v_2 v_3$ and (b) where $v_2 v_4$ and $v_3 v_5$ are non-edges.
Each graph $G$ matching the sequence contains exactly $M_1$ ordered pairs. There are at most $6 m_t$ ordered pairs that contain a triple-edge, at most $4 m_d$ ordered pairs that contain a double-edge and at most $2 m_l$ ordered pairs that contain a loop.
For case (a), there are at most $4 d_1$ ordered pairs where $v_4 \in \{v_2, v_3\}$ or $v_5 \in \{v_2, v_3\}$, and at most $2 d_h$ ordered pairs where $v_4 = v_1$ or $v_5 = v_1$. For case (b), we know that $A_2 = \sum_{i=1}^{d_1} d_i$ is an upper bound on the number of two-paths $v_2 v_4 v_5$ or $v_3 v_5 v_4$ \cite{DBLP:conf/soda/GaoW18}, so there are at most $2 A_2$ such pairs.
\end{proof}

\noindent For Phase 4, we use three new lower bounds $\underline{b}_t(G;0)$, $\underline{b}_t(G;1)$ and $\underline{b}_\tau(G;i+1)$.
\begin{lemma}\label{lemma:lb4}
Let $\mathcal{S}$ be the class of graphs with $m_t$ light triple-edges and $m_d$ light double-edges (and no other non-simple edges).
For all $G \in \mathcal{S}$, all simple three-stars $v_1 v_3 v_5 v_7$ in $G$, and all triplets with $1 \leq i \leq k$ additional pairs $\overline{V}_{i+1}(S) = (v_1 v_3 v_5 v_7 v_2 v_4 v_6 v_8, \dots, v_{6 + 2i - 1} v_{6 + 2i})$ in $G$ that are created by a valid Phase 4 switching $S$, we have
\begin{align}
 \underline{b}_t(\mathcal{S};0) &\leq b_t(G, \emptyset)
 \\
 \underline{b}_t(\mathcal{S};1) &\leq b_t(G, v_1 v_3 v_5 v_7)
 \\
 \underline{b}_\tau(\mathcal{S};i+1) &\leq b_\tau(G, \overline{V}_{i+1}(S)).
\end{align}
\end{lemma}

\begin{proof}
 We have $\underline{b}_t(\mathcal{S};0) = M_3 - 18 m_t d_1^2 - 12 m_d d_1^2$ and $b_t(G, \emptyset)$ is set to the number of simple ordered three-stars in $G$.
Analogously to Lemma 1, we show that $\underline{b}_t(\mathcal{S};0)$ is a lower bound on $b_t(G, \emptyset)$ by starting with $M_3$, the number of ordered three-stars in a graph $G$ matching the sequence and then subtracting an overestimate of the number of non-simple three-stars.
The only non-simple three-stars contain a triple-edge or a double-edge.
There are at most $18 m_t d_1^2$ non-simple three-stars that contain a triple-edge, as there are $m_t$ triple-edges in $G$, at most $d_1$ choices for each of the two remaining outer nodes, and $18$ ways to label the star as shown in \autoref{fig:tswitching}.
Similarly, there are at most $12 m_d d_1^2$ three-stars that contain a double-edge.

 For the second bound, we have $\underline{b}_t(G';1) = L_3 - 18 m_t d_h^2 - 12 m_d d_h^2 - B_3 - 3 (m_t + m_d) B_2 - d_h^3 - 9 B_2$, and $b_t(G, v_1 v_3 v_5 v_7)$ is equal to the number of light simple ordered three-stars that a) do not share any nodes with the three-star $v_1 v_3 v_5 v_7$ created by $S$, b) have no edge $v_1 v_2$ and no multi-edges $v_3 v_4$, $v_5 v_6$, $v_7 v_8$.
Each graph matching the sequence contains exactly $L_3$ light ordered simple three-stars.
Analogous to $\underline{b}_t(\mathcal{S};0)$, there are at most $18 m_t d_h^2 + 12 m_d d_h^2$ light three-stars that are not simple.
There are at most $d_h^3 + 9 B_2$ light simple ordered three-stars $v_2 v_4 v_6 v_8$ of case a): first, if $v_2 = v_1$, then there are at most $d_h^3$ choices for the outer nodes.
In addition, we know that for each node $v_4$ in $G$, there are at most $B_2 = \sum_{i=1}^{d_1} [d_{h+i}]_2$ light simple two-stars $v_2 v_6 v_8$ where $v_2 v_4$ is an edge \cite{DBLP:conf/soda/GaoW18}, so there are at most $9 B_2$ three-stars where $v_4, v_6, v_8 \in \{v_3, v_5, v_7\}$.
The only remaining case is if $v_2 \in \{v_3, v_5, v_7\}$, or if any of $v_4, v_6, v_8 = v_1$, but in this case $v_1 v_2$ is an edge, so this falls under case b). For case b), it suffices to subtract $B_3 + 3 (m_t + m_d) B_2$ three-stars: we know that for each node $v_1$ in $G$, there are at most $B_3 = \sum_{i=1}^{d_1} [d_{h+i}]_3$ light simple three-stars $v_2 v_4 v_6 v_8$ where $v_2 v_1$ is an edge.
For a three-star where any of $v_3 v_4$, $v_5 v_6$, $v_7 v_8$ is a multi-edge, we have at most $3 (m_t + m_d) B_2$ choices, as there are $m_t + m_d$ multi-edges in $G$ and choices for the first outer node, and at most $B_2$ choices for the center and the two remaining outer nodes.

 For the third bound, we have $\underline{b}_\tau(\mathcal{S};i+1) = M_1 - 6 m_t - 4 m_d - 16 d_1 - 4 (i - 1) d_1 - 2 A_2$, and $b_\tau(G, \overline{V}_{i+1}(S))$ is equal to the number of simple ordered pairs in $G$, that a) do not share any nodes with the triplet, or the previous $i - 1$ pairs, and b) have no forbidden edges with the triplet.
First, each graph matching the sequence contains exactly $M_1$ ordered pairs.
At most $6 m_t$ of those pairs are in a  triple-edge, and at most $4 m_d$ pairs are in a double-edge. For case a), there are at most $16 d_1$ ordered pairs that share a node with the triplet, as for each of the $8$ nodes of the triplet, there are at most $d_1$ choices for the second node of the simple pair and $2$ ways to label the pair.
Similarly, there are at most $4 (i - 1)$ pairs that share a node with the $i - 1$ pairs relaxed in the previous steps.
Finally, there are at most $2 A_2$ pairs of case b): each of the two nodes in the pair cannot have an edge with one designated node of the triplet, and starting from that node, there are at most $A_2$ pairs connected to it via an edge.
\end{proof}

\goodbreak
\noindent In Phase 5, we use three new lower bounds $\underline{b}_d(G;0)$, $\underline{b}_d(G;1)$ and $\underline{b}_\tau(G;k+1)$.
\begin{lemma}\label{lemma:lb5}
Let $\mathcal{S}$ be the class of graphs with $m_d$ light double-edges (and no other non-simple edges).
For all $G \in \mathcal{S}$, all simple two-stars $v_1 v_3 v_5$ in $G$, and all doublets with $1 \leq i \leq k$ additional pairs $\overline{V}_{i+1}(S) = (v_1 v_3 v_5 v_2 v_4 v_6, \dots, v_{4 + 2i - 1} v_{4 + 2i})$ in $G$ that are created by a valid Phase 5 switching $S$, we have
\begin{align}
 \underline{b}_d(\mathcal{S};0) &\leq b_d(G, \emptyset)
 \\
 \underline{b}_d(\mathcal{S};1) &\leq b_d(G, v_1 v_3 v_5)
 \\
 \underline{b}_\tau(\mathcal{S};i+1) &\leq b_\tau(G, \overline{V}_{i+1}(S)).
\end{align}
\end{lemma}

\begin{proof}
We have $\underline{b}_d(G';0) = M_2 - 8 m_d d_1$, and $b_d(G, \emptyset)$ is set to the number of simple ordered two-stars in $G$.
We now show that $\underline{b}_d(G';0)$ is a lower bound on $b_d(G, \emptyset)$.
There are $M_2$ ordered two-stars in a graph $G$ matching the sequence.
Of these, the only ones that are not simple are the ones that contain a double-edge, and $G$ can contain at most $8 m_d d_1$ such two-stars.

For the second bound, we have $\underline{b}_d(G;1) = L_2 - 8 m_d d_h - 6 B_1 - 3 d_h^2$, and $b_d(G, v_1 v_3 v_5)$ is equal to the number of light simple ordered two-stars in that do not share any nodes with the two-star $v_1 v_3 v_5$.
Similar to the first step above, $G$ contains exactly $L_2$ light ordered two-stars, and at most $8 m_d d_h$ light ordered two-stars that are not simple.
The only remaining cases are two-stars $v_2 v_4 v_6$ that share any nodes with the first two-star.
First, there are at most $6 B_1$ two-stars where $v_4, v_6 \in \{v_1, v_3, v_5\}$, as $B_1$ is an upper bound on the number of pairs $v_2 v_4$ or $v_2 v_6$ where $v_4$ or $v_6$ are connected to one of $v_1, v_3, v_5$ via an edge.
The other remaining case is $v_2 \in \{v_1, v_3, v_5\}$. In this case, there are at most $d_h^2$ choices for the remaining nodes of the two-star, so in total there are at most $3 d_h^2$ such two-stars.

The proof for $\underline{b}_\tau(\mathcal{S};i+1)$ is analogous to the proof for the similar bound in Phase 4 (see above).
\end{proof}

\end{document}